\documentclass[12pt,a4paper]{amsart}

\usepackage[T1]{fontenc}
\usepackage[utf8]{inputenc}
\usepackage[english]{babel}
\usepackage{amsfonts, amssymb, amsmath, amsthm}
\usepackage[all]{xy}
\usepackage{enumerate}

\usepackage{graphicx}
\usepackage{psfrag}
\usepackage{geometry}
\usepackage[colorlinks=true,pdfstartview=FitV,linkcolor=blue,citecolor=green,urlcolor=black,filecolor=magenta]{hyperref}
\usepackage{verbatim}

\newtheorem{theorem}{Theorem}[section]
\newtheorem{definition}[theorem]{Definition}
\newtheorem{prop}[theorem]{Proposition}
\newtheorem{lemma}[theorem]{Lemma}
\newtheorem{cor}[theorem]{Corollary}

\newcommand{\Aa}{{\mathcal A}}
\newcommand{\Bb}{{\mathcal B}}

\newcommand{\Dd}{{\mathcal D}}

\newcommand{\Ll}{{\mathcal L}}
\newcommand{\Mm}{{\mathcal M}}

\newcommand{\Pp}{{\mathcal P}}

\newcommand{\Tt}{{\mathcal T}}

\newcommand{\id}{\mathrm{id}}

\newcommand{\NM}{{\mathbb N}}

\newcommand{\Z}{{\mathbb Z}}

\newcommand{\R}{{\mathbb R}}

\newcommand{\C}{{\mathbb C}}

\newcommand{\Tr}{\mathrm{Tr}}
\newcommand{\ch}{\mbox{\rm ch}}

\newcommand{\tp}{\mathfrak t}

 \newcommand{\gap}{{\mathcal Gap}_\pe}
 \newcommand{\gapS}{{\mathcal Gap}_{\pe,\se}}
 \newcommand{\gapLS}{{\mathcal Gap}_{\pe,\LS}}

 \newcommand{\pe}{\Ll}
 \newcommand{\LS}{\underline{\pe}}
 \newcommand{\Et}{\mathcal E^{top}_\pe}
\newcommand{\mean}{\mathfrak{m}}
\newcommand{\lang}{L}
\newcommand{\al}{\mathfrak a}
\newcommand{\se}{S}
\newcommand{\Ch}{\mathrm{Ch}}
\newcommand{\Ape}{\Aa_\pe}
\newcommand{\Atb}{\Aa_{\pe,\se}}
\newcommand{\ttb}{{\tp_\se}}

\title{Bragg spectrum, $K$-theory and Gap Labelling of aperiodic solids}

\author{Johannes Kellendonk}
\address{Institut Camille Jordan, Universit\'{e} Lyon~1, 69622 Villeurbanne, France}
\email{kellendonk@math.univ-lyon1.fr}

\date{Version of \today}

\begin{document}
\maketitle
\begin{abstract}
The diffraction spectrum of an aperiodic solid is related to the group of eigenvalues of the dynamical system associated with the solid. Those eigenvalues with continuous eigenfunctions constitute the topological Bragg spectrum. We relate the topological Bragg spectrum to the topological invariants (Chern numbers) of the solid and to 
the gap-labelling group, which is the group of possible gap labels for the spectrum of a Schr\"odinger operator describing the electronic motion in the solid.       
\end{abstract}


\section{Introduction}
A longstanding question in solid state physics is, how is the X-ray diffraction spectrum of a solid related to its electronic spectrum? For crystalline solids the answer is known since the early days of quantum mechanics: X-rays with wave vectors corresponding to Bragg peaks perturb the electronic motion and create gaps in the electronic spectrum \cite{ashcroft1976solid,reed1978iv}. 
For aperiodic solids this question is more subtle.
A perturbative expansion of the Liapunov exponent of the transfer matrix dynamical system of one dimensional tight binding models combined with heuristic arguments has led to a conjecture about where gaps in the spectrum can be found depending on the diffraction measure \cite{luck1989cantor}, but we are not aware of any rigorous development of this approach. Furthermore,
for quasi-crystals and other long range ordered structures approximation methods have been employed to obtain results resembling the periodic case  \cite{kaliteevski2000two,gambaudo2014brillouin,dareau2017revealing}. We approach the question asked above from a topological point of view establishing a relation between the two kinds of spectra which depends only on the spatial structure of the solid. By this spatial structure, which we denote $\pe$, we mean the set of positions of the atoms in the solid\footnote{equivalently, tilings could be used to describe the spatial structure}  and their type. Its topological properties can be described by a $C^*$-algebra $\Ape$ and a dynamical system associated to $\pe$. In this framework we establish in particular a link between the Bragg spectrum of the solid and a topological invariant associated with its electronic spectrum, namely its gap-labelling group. 

One might ask at this point, in which sense is diffraction topological at all? Indeed, a first glimpse at the definition of the diffraction measure shows that it does not change, if one modifies the underlying structure on a set of vanishing density, and therefore diffraction theory has more the flavor of a measure than a topological theory. In contrast to that, the topology of the solid can change quite a bit if the structure is modified on a set of density zero. We therefore first have to share out the topological content of diffraction theory. It is here where a third notion of spectrum comes into play: the dynamical spectrum and more specifically the topological eigenvalues of the dynamical system associated to the solid. One of the major results in diffraction theory states that the positions of the Bragg peaks are eigenvalues of the dynamical system of the solid. In the case of crystals, quasi-crystals and many more structures studied in aperiodic order theory these eigenvalues are topological in the sense that they correspond to wave-vectors $k$ for which the plane wave $x\mapsto e^{i\langle k,x\rangle}$ is pattern equivariant \cite{kellendonk2013topological,kellendonk2014meyer}. Pattern equivariance can be understood as a coherence condition: whenever the solid looks around $x$ the same as around $y$ out to a large distance but perhaps up to a small error, then the plane wave has the same phase at $x$ as at $y$, up to a small error.
It can be regarded as a sign of high long range order of a material if most, if not all, of its  eigenvalues are topological. In an analogy, this is like saying that materials with high long range order correspond to continuous representatives of elements in an $L^2$-theory; and in this sense their diffraction is topological. 

Let $\Et$ stand for the group formed by the topological Bragg peaks. We will construct a morphism $\Phi:\Lambda^n\Et \to K_{n+d}(\Ape)$ from the exterior module of $\Et$ into the $K$-theory of the algebra of the $d$ dimensional solid. It tells us that a good part of the $K$-theoretic invariants of the solid come from its topological Bragg spectrum. We will furthermore see that the Chern numbers of the invariants coming from the topological Bragg spectrum can be directly obtain from a simple module morphism from $\Lambda^n\Et$ to $\Lambda^{d-n}\R^d$.
This can be summarized by the commutative diagram
\begin{equation}\label{eq-main-diagram}
\begin{array}{ccccc}
 \Lambda^n\Et & \stackrel{\Phi}\to  & K_{n+d}(\Ape ) \\
 \downarrow &           & \downarrow \\ \langle \ch^{(d-n)},\cdot\rangle  \\       
 \Lambda^n{\R^d}^* &  \stackrel{D}\to   &   \Lambda^{d-n}{\R^d}
 \end{array}
\end{equation}
where the left vertical map is the morphism induced by the inclusion $\imath:\Et\to{\R^d}^*$, $D$ is the Poincar\'e (or Hodge-) duality operator, and the right vertical map the Connes pairing with the standard Chern-cocycles defined by the translation action and its dual on $\Ape$.

The above is formulated for the cases in which there is no external magnetic field. We leave the question of what happens under the influence of a magnetic field open, commenting only on the two-dimensional case, which already indicates that the situation is more subtle.  

A special case arises if $n$ equals the dimension $d$ of the solid in the above diagram. In this case the right vertical arrow corresponds to the linear functional $\tp_*$ on $K_0(\Ape)$ induced by the trace per unit volume which we will denote here $\tp$.  The above diagram is then related to Bellissard's $K$-theoretical formulation of the gap labelling and the image of the right vertical arrow $\tp_*$ is the gap-labelling group \cite{bellissard1986k,bellissard1992gap}. Indeed, (\ref{eq-main-diagram}) becomes
 $$\begin{array}{ccccc}
 \Lambda^d\Et & \stackrel{\Phi}\to  & K_{0}(\Ape ) \\
 \downarrow \frac1{(2\pi)^d}\det&           & \downarrow  \tp_*  \\       
 \R^* & \cong & \R 
 \end{array}$$
where
$\det(k_1\wedge\cdots \wedge k_d) $ is the determinant of the matrix spanned by the vectors $k_1,\cdots, k_d$, that is, the signed volume of the parallel epiped spanned by these vectors.
The composition $\tp_*\circ \Phi$, which we also call the Bragg to Gap map, is  in one dimension given by $k\mapsto \frac{k}{2\pi}$ and hence injective. 
We investigate the question of when the image of $\tp_*\circ \Phi$ exhausts the gap-labelling group. Our answer provides a justification of the gap-labelling found in the literature for quasi-periodic chains, but also points out the limitation of our topological approach. Indeed, if the Bragg to Gap map is not surjective, then there are gap-labels which to not come from the Bragg spectrum. This happens in particular for the Thue-Mores chain and so the mystery of the gap-opening mechanism for the Kohmoto Hamiltonian on the Thue-Morse chain escapes also our topological analysis.

A recent article \cite{akkermans2021relating}  turns around similar questions. Its main part is about the comparison between the $K$-theory of $\Ape$ and various cohomology theories which can be associated to the solid, and how the trace per unit volume can be formulated in cohomology. We discuss this result here in Section~\ref{sec-coh}. What our work adds to the picture of \cite{akkermans2021relating} is the relation between the topological Bragg spectrum and cohomology thus providing the theoretical background for the observations made in Section~11 of \cite{akkermans2021relating}. We would also like to mention the work by Itza-Ortiz \cite{itza2007eigenvalues}, in which an equivalent formulation of Theorem~\ref{thm-diagram-1} for $d=1$ 
has been obtained in the context of 
dynamical systems and their 
dimension groups.

The article is organised as follows: In Section~\ref{sec-2} we review the description of the spatial structure of a solid by means of decorated point patterns. From these point patterns the pattern algebra is derived. We follow here the description of \cite{Kellendonk2020} which develops that of \cite{kellendonk2000tilings} and is dual to that of \cite{bellissard1992gap}, in particular the pattern space (or hull) arises as the Gelfand spectrum of the pattern algebra. 
After a very short summary of the definition of $K$-theory we present 
in Section~\ref{sec-gap} the gap-labelling of Bellissard. 
In Section~\ref{sec-Bragg} we describe mathematical diffraction theory and recall the all important relation between the diffraction spectrum and the dynamical spectrum of a solid. In particular, the Bragg peaks of the diffraction spectrum are related to the topological eigenvalues. 
In Section~\ref{sec-main} we relate the topological eigenvalues to the $K$-theory of the pattern algebra. The main result (the detailed version of (\ref{eq-main-diagram})) is Theorem~\ref{thm-main-diagram}. 
In Section~\ref{sec-coh} we compare our approach to that of \cite{akkermans2021relating} and discuss in particular the relation between
the topological eigenvalues and the cohomology of the pattern space of the solid. Finally we discuss the questions about the bijectivity of the Bragg to Gap map and explain the example of the Thue-Morse chain in Section~\ref{sec-one-dim}.

\section{The spatial structure of the solid}\label{sec-2}
The solid is an arrangement of atoms (or ions) of different types which create a background potential for the electrons and an electronic density at which the $X$-ray beams are scattered in a diffraction experiment.
What we refer to as the spatial structure of the solid is the set of positions of the atoms, the local configurations or clusters of atoms and how they repeat in space. 
This can be modelled through a decorated Delone set in $\R^d$. We treat the solid as infinitely extended in all directions, so there is no boundary.

Let $ \al$ be a (finite) set of symbols encoding the different type of atoms (ions) of the solid. Let $\pe_a\subset\R^d$ be the set of positions of atoms of type $a$. We assume that the set $\LS:=\bigcup_{a\in \al} \pe_a$ is a Delone set (a uniformly discrete, relatively dense set). The collection $\pe$ of all $\pe_a$ is called a decorated Delone set of $\R^d$. We may identify $\pe$ with the subset $\prod_{a\in \al}\pe_a\subset \prod_{a\in \al}\R^d$. 
It describes the spatial structure of the solid. Following \cite{Kellendonk2020} (in which the constructions of \cite{kellendonk1995noncommutative,kellendonk2000tilings} are generalised beyond the finite local complexity case) we derive from $\pe$ a $C^*$-algebra $\Ape$ which can be seen as the observable algebra of the structure.
All topological results below will depend only on $\pe$. 

The $R$-patch (or cluster)  of $\pe$ at $x\in \R^d$ is 
$$B_R[\pe;x] := \prod_{a\in \al} (B_R(x)\cap \pe_a) \cup \partial B_R(x).$$
Here $B_R(x)$ is the ball of radius $R$ centered at $x\in\R^d$ and $\partial B_R(x)$ the sphere. 
The group of translations acts (diagonally) on $\R^d\times\cdots\times\R^d$. We define an $R$-patch class to be an equivalence class of an $R$-patch under this translation action. We say that the  $R$-patch class $\mathcal P$ occurs at $x\in \R^d$ if $B_R[\pe;x] \in \mathcal P$.

\subsection{Pattern equivariant functions and operators}\label{sec-2.1}
Let $\lang^R(\pe)_0$ be the set of $R$-patch classes of $\pe$ equipped with the Hausdorff metric topology: the class of $B_R[\pe;x]$ is at least $\epsilon$-close to the class of  $B_R[\pe;y]$ if there is $z\in \R^d$ such that 
any point of $B_R[\pe;x]$ is within distance $\epsilon$ of a point of $B_R[\pe;y]+z$ and vice versa.
We then define $\lang^R(\pe)$ to be the closure of $\lang^R(\pe)_0$ in that topology. 
\begin{definition}\label{def-PEF}
We say that a bounded continuous function $f:\R^d\to\C$ is strongly pattern equivariant with radius $R>0$ if there is a continuous function $b:\lang^R(\pe)\to\C$ such that $f(x) = b(B_R[\pe;x])$ for all $x\in\R^d$. 
We define $C_\pe(\R^d)$ to be the closure of all strongly pattern equivariant functions in the sup norm.
\end{definition}
An example of a strongly pattern equivariant function is the convolution $g\ast \delta_{\pe_a}$ of a compactly supported continuous function $g:\R^d\to \C$ with the sum of the Dirac measures at the points of $\pe_a$,
$\delta_{\pe_a}=\sum_{x\in\pe_a} \delta_x$, $\delta_x(A) = 1$ if $x\in A\subset\R^d $ but $0$ if $x\notin A$. 

We call the functions of $C_\pe(\R^d)$ pattern equivariant and note that they can be described as follows: strongly pattern equivariant functions with radius $R$ form a unital sub-algebra of all bounded continuous functions on $\R^d$ and a function which is strongly pattern equivariant with radius $R$ is also strongly pattern equivariant with radius $R'>R$. Hence strongly pattern equivariant functions define a directed system of algebras and $C_\pe(\R^d)$ is its $C^*$-direct limit. 
The Gelfand spectrum of $C_\pe(\R^d)$
is hence the inverse limit of the topological spaces $\lang^R(\pe)$. This space, which we denote $\Omega_\pe$, is usually referred to as the hull of $\pe$. The action $\alpha:\R^d \to \mathrm{Aut}(C_\pe(\R^d))$ by translation, $\alpha_a(f) (x) = f(x+a)$, induces an $\R^d$-action on $\Omega_\pe$. 

The double $R$-patch of $\pe$ at $(x,y)\in \R^d\times\R^d$ is the subset
$$B_R[\pe;x,y]:=B_R[\pe;x] \times B_R[\pe;y]\subset  \left(\prod_{a\in \al} \R^d \right) \times \left(\prod_{a\in \al} \R^d \right).$$
We call $|y-x|$ the range of the double patch. We use again the diagonal action to translate double $R$-patches and to define a double $R$-patch class as the equivalence class of the double patch under translation. 

Let ${\lang^{R,M}}_0$ be the set of double $R$-patch classes of $\pe$ of range $\leq M$ equipped with the analogous metric topology as above. Again we let $\lang^{R,M}$ be the closure of ${\lang^{R,M}}_0$.
\begin{definition}\label{def-PEK}
We say that a bounded continuous function $F:\R^d\times\R^d\to\C$ is strongly pattern equivariant with radius $R>0$ and range $M> 0$ if there is a continuous function $b:\lang^{R,M}\to\C$ 
satisfying 
\begin{equation}\label{eq-M}
\lim_{|x-y|\to M} b(B_R[\pe;x,y])=0
\end{equation} and
such that $F(x,y) = b(B_R[\pe;x,y])$ for all $x,y\in\R^d$. 
\end{definition}
The space of continuous functions $b:\lang^{R,M}\to\C$ satisfying (\ref{eq-M})
is naturally included in the  space of continuous functions $b:\lang^{R',M'}\to\C$ 
satisfying (\ref{eq-M}) with $M'\geq M$, 
provided $R'\geq R$. We may therefore define
the direct limit over $R\to+\infty$ and $M\to \infty$ of all strongly pattern equivariant functions with radius $R>0$ and range $M > 0$. We denote this vector space by $\Ape^{(s)}$. Its elements are pattern equivariant integral operators. 
We equip $\Ape^{(s)}$ with
the product and involution
\begin{equation}\label{eq-prod} 
F_1 F_2 (x,y) = \int_{\R^d} F_1(x,z) F_2(z,y)dz,\quad F^*(x,y) = \overline{F(y,x)}.
\end{equation}
Here, integration is w.r.t.\ Lebesgue measure on $\R^d$.
We define $\Ape$ to be the closure of $\Ape^{(s)}$ in the norm given by the representation $\pi$ on $L^2(\R^d)$,
$$\pi(F)\psi(x) = \int_{\R^d} F(x,y)\psi(y) dy.$$
Whenever convenient, we identify $\Ape$ and $\Ape^{(s)}$ with operators in this representation. We call $\Ape$ the continuous pattern algebra. It may be seen as the algebra of observables of the solid.  

Let $\alpha:\R^d \to \mathrm{Aut}(C_\pe(\R^d))$ be the translation action $\alpha_a(f) (x) = f(x+a)$. Recall that the crossed product $C_\pe(\R^d)\rtimes_\alpha \R^d$ is the $C^*$-closure of the algebra 
$C_c(\R^d,C_\pe(\R^d))$, of continuous compactly supported functions $\tilde F:\R^d\to C_\pe(\R^d)$ with $\alpha$-twisted convolution product and involution,
$$ \tilde F_1 \ast_\alpha \tilde F_2(h) = \int_{\R^d} \tilde F_1(x) \alpha_x(\tilde F_2(h-x)) dx,\quad \tilde F^*(h) = \alpha^{}_h(\tilde F(-h)^*).$$
\begin{lemma}
$\Ape$ is isomorphic to the crossed product $C_\pe(\R^d)\rtimes_\alpha \R^d$. 
\end{lemma}
\begin{proof} Denote by $C^{(s)}_\pe(\R^d)$ all strongly pattern equivariant functions. 
The map 
$\Ape^{(s)}\ni F \mapsto \tilde F\in C_c(\R^d,C^{(s)}_\pe(\R^d))$,
$$\tilde F(h)(x) = F(x,x+h) $$ is easily seen to be bijective and to preserve the product and the involution. 
$C_\pe(\R^d)\rtimes_\alpha \R^d$ is the closure in the norm of the induced representation of the representation of $C_\pe(\R^d)$ on $L^2(\R^d)$ be left multiplication. This induced representation is unitarily equivalent to the infinite direct sum of the representation $\pi$. The norm defining the closure to obtain $C_\pe(\R^d)\rtimes_\alpha \R^d$ is therefore the same as the norm defining the closure to obtain $\Ape$.\end{proof}

\subsection{Trace per unit volume}
The operator trace on operators on $L^2(\R^d)$ is undefined on most of the elements of $\Ape$, one has to take it per unit volume in order to get finite values. However, the construction of a trace per unit volume is not always unambiguous and may depend on choices. 

Let  $(\Lambda_n)_n$ be a van Hove sequence. 
The sequence of linear functionals $\mean_n: C_\pe(\R^d)\to \C$,
$$ \mean_n(f) = \frac1{\mathrm{vol}(\Lambda_n)}\int_{\Lambda_n} f(x) dx$$
admits a converging subsequence in the weak-*-topology. By going over to a sub-sequence we may assume that $$\mean(f) = \lim_n \mean_n(f)$$
exists
for all $f\in C_\pe(\R^d)$. It then defines a positive linear functional on $C_\pe(\R^d)$ of norm $1$ which is easily seen to be translation invariant. It can be used to define the frequencies of patches. Indeed, the frequency of an $R$-patch class $\mathcal P$ is the density of the set  $\pe_\mathcal P= \{x\in \R^d : B_R[\pe;x] \in \mathcal P\}$ of points where it occurs, and this density is determined via $\mean$: Given any continuous function $\kappa:\R^d\to \R$ of compact support and integral $1$, the frequency of $\Pp$ is $\mean(\kappa\ast \delta_{\pe_\Pp})$.  
If all van Hove sequences define the same frequencies of its $R$-patches then the frequency of $\Pp$ is given by the naive formula: count the points in $B_R(x)\cap \pe_\Pp$ devide by the volume of $B_R(x)$ and consider the limit when $R\to +\infty$.

The mean $\mean$ on $C_\pe(\R^d)$ extends to a semi-finite, lower semi-continuous trace on $\Ape$ which we denote by $\tp$, it is given on $\Ape^{(s)}$ by
$$\tp (F) = \mean(x\mapsto F(x,x))$$
and thus corresponds to the trace per unit volume. In the interpretation of $\Ape$ as a crossed product, 
$\tp$ is known as  the dual trace $\tp=\hat\mean :C_\pe(\R^d)\rtimes_\alpha \R^d\to \C$.

 \subsection{Discrete pattern algebras}
 To describe the physical operators in the tight binding approximation one uses a discrete version of the tiling algebra. Here the word discrete refers to the fact that the translation action of $\R^d$ is broken down to a discrete set of translations.

 A uniformly discrete subset $\se$ of $\R^d$ is locally derivable from $\pe$ if there exists $R>  0$ and a continuous function $b:\lang^R(\pe) \to \lang^1(\se)$ such that, for all $x$, the $1$-patch class of $\se$ at $x$ is $b$ applied to the $R$-patch class of $\pe$ at $x$, $B_1[\se;x] = b(B_R[\pe;x])$.
 A simple example is $S=\LS$ where the map $b$ simply forgets the symbols but below we allow $\se$ to be any Delone set of $\R^d$ which is locally derivable from $\pe$.
 
By restricting in the above definitions \ref{def-PEF} and \ref{def-PEK} the pattern equivariant functions or integral kernels to $S$ or to $S\times S$, we obtain the discrete algebra $C_\pe(S)$ of pattern equivariant functions $f:S\to\C$ and the space $\Atb^{(s)}$ of pattern equivariant kernels $F:S\times S\to\C$. We equip $\Atb^{(s)}$ with the product and the involution
\begin{equation}\label{eq-prod-s} F_1 F_2 (s,t) = \sum_{r\in \se} F_1(s,r) F_2(r,t),\quad F^*(s,t) = \overline{F(t,s)},
\end{equation}
and let $\Atb$ be the completion in the norm defined by the representation $\pi_S$ on $\ell^2(S)$,
$$\pi_S(F)\psi(s) =  \sum_{r\in \se} F(s,r) \psi(r) .$$
We call $\Atb$ the discrete pattern algebra associated to $\se$.  The Gelfand spectrum of $C_\pe(S)$ is the discrete pattern space of $\pe$ associated $\se$,  we denote it by $\Xi_{\pe,\se}$. As $\se$ is relatively dense, $\Xi_{\pe,\se}$ is a transversal to the $\R^d$-action on $\Omega_\pe$. $\Xi_{\pe,\LS}$ is also referred to as the canonical transversal of $\Omega_\pe$.
 
In general, $\Atb$ is a groupoid $C^*$-algebra. If $d=1$ then $\Atb$ is isomorphic to the crossed product of $C_\pe(S)$ by $\Z$, where the $\Z$-action is given by $\alpha_1(f)(s) = f(s')$ with $s'\in \se$ being the next point to the right of $s$.

\subsection{Morita equivalence}
\label{sec-ME}
\newcommand{\ME}{\Mm}
Topological effects in solids, like their topological phase and the bulk boundary correspondence, can be described in both, the continuous system or in the tight binding approximation. This comes about as 
the continuous and the discrete pattern algebra are (strongly) Morita equivalent and therefore have the same $K$-theory. We explain this adapting the exposition of \cite{bellissard1992gap} to our framework.

The notion of pattern equivariance can be applied to bounded continuous functions 
$\xi:\se\times\R^d\to \C$, leading a Morita equivalence $\Atb-\Ape$-bimodule:
We define first the linear space $ \ME^{(s)}$ of functions $ \xi:\se\times\R^d\to \C$ which are restrictions to 
$\se\times\R^d\subset \R^d\times\R^d$ of functions from $\Ape^{(s)}$. This space carries a left 
$\Atb^{(s)}$ action $\Atb^{(s)}\times \ME^{(s)}\ni (F,\xi) \mapsto F\xi \in \ME^{(s)}$ by the analog of formula (\ref{eq-prod-s}) (with $F$ in place of $F_1$ and $\xi$ in place of $F_2$). It carries also a right $\Ape^{(s)}$ action $\ME^{(s)}\times \Ape^{(s)}\ni (\xi,F) \mapsto \xi F \in \ME^{(s)}$ by the analog of formula (\ref{eq-prod}) (with $\xi$ in place of $F_1$ and $F$ in place of $F_2$) It carries moreover an $\Atb$-valued  and an $\Ape$-valued scalar product which are given again by formulas analogous to (\ref{eq-prod}) and (\ref{eq-prod-s})
$$( \xi_1, \xi_2)_{\Atb} =  \xi_1  \xi_2^*,  \qquad ( \xi_1, \xi_2)_{\Ape}={ \xi_1^*}  \xi_2$$
where $ \xi^*(x,s) := \overline{ \xi(s,x)}$.
Then 
$ \|  \xi \|^2_ \ME = \|( \xi^*, \xi)\|_{\Ape} = \|( \xi,  \xi^*)\|_{\Atb}$
defines a norm on $ \ME^{(s)}$. Its completion in that norm $\ME$ is a  $\Atb-\Ape$-bimodule. 

Let $u\in  \ME$ with $(u,u)_{\Atb} = 1$. Then $\Pi_u = (u,u)_{\Ape}$ is an orthogonal projection.
Furthermore, we obtain an injective $*$-algebra morphism
$i_u:\Atb\to \Ape$ through
$$i_u(c) = (u,cu)_{\Ape}.$$
As $S$ is relatively dense so that $\Xi_{\pe,\se}$ is a transversal in $\Omega_\pe$, $\ME$ is a Morita equivalence bimodule and $i_u$ 
induces an isomorphism ${i_u}_*:K_i(\Atb)\to K_i(\Ape)$ between the $K$-groups.
As $\se$ is uniformly discrete, it is not difficult to construct a function $u$ with the above properties. If $r_{min}$ is the separation constant of $\se$, that is, $r_{min} = \inf_{s\neq t\in\se} |s-t|$, any continuous function $g:\R^d\to \C$ with support contained in the $\frac{r_{min}}2$-ball around $0$ and $\int |g(x)|^2 dx =1$ yields such a $u$, notably $u(s,y) := g(y-s)$. With that choice
$$ (u,u)_{\Atb}(s,t) = \int u(s,z) \overline{u(t,z)} dz = \int g(z) \overline{g(z+s-t)} dz= \delta_{s,t}$$
for $s,t\in\se$ 
and 
$$\Pi_u(x,y) = \sum_{s\in \se} \overline{u(s,x)} u(s,y) =  \sum_{s\in \se} \overline{g(x-s)} g(y-s) .$$

The trace per unit volume $\tp$ on $\Ape$ induces a finite trace $\ttb$ on $\Atb$, namely
\begin{equation} \label{eq-tr}
\ttb = :\nu \, \tp\circ i_u
\end{equation}
where 
$$\nu^{-1} = \tp(\Pi_u) = \mean(|g|^2\ast \delta_\se) = \mathrm{dens}(\se).  $$
This situation can be summarised in the commuting diagram
\begin{equation}\label{eq-ME}
\begin{array}{ccc}
K_{0}(\Ape) & \stackrel{{i_u}_*^{-1}}\to &  K_{0}(\Atb)\\
 \downarrow \tp_* &      & \downarrow  \mathrm{dens} S \,\ttb_*  \\       
 \R & = & \R 
 \end{array}
 \end{equation}
While the inclusion of $\Atb$ in $\Ape$ depends on the choice of $u$, the trace $\ttb$ does not depend on it. Note that $\ttb$ is normalised, $\ttb(1) = 1$.  
 \bigskip
 
We say that $\pe$ has \emph{finite local complexity} if, for any $R>0$, there are only finitely many classes of $R$-patches  $B_R[\pe;x]$ with $x\in \LS $. 
This condition implies that $\Xi_{\pe,\LS}$ is a totally disconnected space and $\Ape$ strongly Morita equivalent to a crossed product algebra $C(\Xi)\rtimes_\alpha\Z^d$ of a $\Z^d$-action on a totally disconnected compact space $\Xi$ so that one may compute the $K$ theory of $\Ape$ with the help of the Kasparov spectral sequence.

\subsection{Schr\"odinger operators in the one-particle approximation}
In the one-particle approximation, the motion of a particle in the solid is described by the Schr\"odinger equation with an unbounded Hamiltonian $H$ on $L^2(\R^d)$. It can therefore not belong to the pattern algebra $\Ape$. Instead, we replace $H$ by a bounded operators which contains the same topological information. Again, we adapt  \cite{bellissard1992gap} to our framework.

Let $H_0$ be a translation invariant closed self adjoint operator whose resolvent $(H_0+i)^{-1}$ is an integral operator on $L^2(\R^d)$. We may take the free Laplacian $H_0=-\Delta$, for example. $(H_0+i)^{-1}$ is pattern equivariant and thus an element of $\Ape$. Let $V\in C_\pe(\R^d)$ be real valued. As the Neumann series associated to the resolvent identity is norm-convergent, also the resolvent $(H+i)^{-1}$ of $H=H_0+V$ belongs to $\Ape$. $(H+i)^{-1}$ is a normal element of $\Ape$. By Gelfand theory, $\tilde f((H+i)^{-1})\in\Ape$ for any continuous function $\tilde f:\C\to\C$ which vanishes at $0$ (if $\tilde f(0)\neq 0$ then 
$\tilde f((H+i)^{-1})$ belongs to the unitization of $\Ape$). 
Let $g(z) = z^{-1}-i$ and $f$ be a continuous function $f:\C\to \C$ which vanishes at infinity. Then $f\circ g$ is a continuous function which vanishes at $0$ and therefore $f\circ g((H+i)^{-1})\in\Ape$. This shows that $f(H)\in\Ape$ for any continuous function $f:\R\to \C$ which vanishes at infinity. 

The simplest example of a Hamiltonian which is affiliated to $\Ape$ is
$$ H = -\Delta + \sum_{a\in\al} \sum_{x\in\pe_a} v_a \ast \delta_x. $$

Internal degrees of freedom like spin may be taken into account by going over to the Hilbert space $L^2(\R^d)\otimes\C^N$ and the algebra $M_N(\Ape)$.

The above approach can be extended to external magnetic fields by incorporating the magnetic field in the algebra by means of twisting cocycles \cite{bellissard1992gap}. We focus here on the situation without external magnetic fields.

\subsection{Tight binding operators in the one-particle approximation}
Given a Delone set $S$ of $\R^d$, a
tight binding operator on $S$ is an operator on $\ell^2(S)$ of the form
$$ H\psi(s) = \sum_{s'} H_{s s'} \psi(s')$$
where $H_{s s'}:S\times S\to \C$ is a kernel which decays sufficiently fast if $|s-s'|\to +\infty$. Often even $H_{ss'}=0$ if $|s-s'|$ exceeds a finite range. In our context we assume that $S$ is locally derivable from $\pe$  and that the kernels are pattern equivariant so that $H$ is an element of $\Atb$. The advantage of tight binding operators over Schr\"odinger operators is that they are bounded and easier accessible by numerical methods, but from the topological point of view, they are equivalent.

Internal degrees of freedom like spin may again be taken into account by tensoring on a copy of $\C^N$, $\ell^2(S)\otimes\C^N$ and passing to the algebra $M_N(\Atb)$. 
\section{$K$-theoretical gap labelling}\label{sec-gap}
\subsection{$K$-theory}
We only review the bare minimum to set up the notation and allow for an interpretation of the elements of the $K$-groups sufficient to understand the $K$-theoretical gap-labelling, refering the reader to \cite{blackadar1998k,rordam2000introduction} for details.

The $K_0$-group of a $C^*$-algebra $A$ is obtained from the (orthogonal) projections in $M_n(A^+)$, the algebra of $n\times n$ matrices with coefficients in the algebra obtained by adjointing a unit to $A$. If $A$ has a unit, it is not necessary to use $A^+$, but the continuous pattern algebra $\Ape$ is not unital. The process of adding a unit comes with a so-called scalar map $s:A^+\to\C$, an algebra morphism which maps the added unit to $1\in\C$ and whose kernel is $A$. By identifying an element of $M_n(A^+)$ as the upper left block of a matrix in $M_m(A^+)$, $m>n$, with otherwise $0$ entries, one can build the semigroup
$$\Dd(A^+) = \bigcup_{n\geq 1} \mathrm{Proj}(M_n(A^+))/\sim_h $$ 
where $  \mathrm{Proj}(M_n(A^+))$ denotes the (orthogonal) projections of $M_n(A^+)$ and we quotient out by homotopy. The addition is given by the direct sum of representatives and turns out to be abelian. Applying the Grothendieck functor to $\Dd(A^+)$ , which turns an abelian semi-group into the group of its formal differences, we obtain an abelian group which is denoted $K_{0}(A^+)$. The construction is functorial and so there is a group morphism $K_0(s):K_{0}(A^+)\to \Z$ which maps the class of the projection given by the added unit to $1\in\Z$ and whose kernel is, by definition, $K_0(A)$. I turns out that any element of $K_0(A)$ can be written as a formal difference $[p]_0-[s(p)]_0$ where $p$ is a projection in $M_n(A^+)$, its scalar part $s(p)$ a projection in $M_n(\C)$ and $[\cdot]_0$ denotes their homotopy class. 
 
The construction of $K_1(A)$ is similar, but one considers unitaries instead of projections. Now a unitary in $U_n(A^+)$, the set of unitary elements of $M_n(A^+)$, can identified with a unitary in $U_m(A^+)$, $m>n$ by putting it into the upper left corner and completing the matrix with $1$'s on the diagonal and $0$'s everywhere else. This allows again to take the union and $K_1(A)$ is that union
$$K_1(A) = \bigcup_{n\geq 1} U_n(A^+)/\sim_h $$
which an abelian group under direct sum of the representatives (or, equivalently, their multiplication). 
\subsection{Gap-Labelling} 
A gap in the spectrum of the Hamiltonian $H$ is a connected component of its complement. Since there are at most countably many gaps, the gaps of $H$ could be labeled by natural numbers. But there is more structure to it. Gaps are ordered on the energy line and we want a labelling which respects this order. Furthermore, there is a hidden group structure behind the gaps which comes to light through the use of $K$-theory. The gap labelling group is  a subgroup $\gap$ of $\R$, which depends only on the spatial structure of the solid $\pe$, and is such that the gaps of any Hamiltonian which is resolvent affiliated to $\Ape$ and bounded from below can be labelled by elements of $\gap$ in a way which respects the order on the energy line. Bellissard \cite{bellissard1986k} proposed to use the $K$-theory of $\Ape$ to define the gap labelling group as follows.

If $E$ lies in gap  of the spectrum of $H$ then the spectral projection of $H$ onto the energy interval $(-\infty,E]$ is a bounded continuous function of $H$, namely it is given by $P_E(H)$ for any (continuous) function $P_E$ which is $1$ on $(E_{min},E]$ and $0$ on all energies above the gap. Here $E_{min}$ is any lower bound to $H$ (or $-\infty$). Hence if $H$ is bounded from below and resolvent affiliated then $P_E(H)$ is a projection in $\Ape$, 
as we may take for $P_E$ a continuous function which vanishes at infinity. Clearly $P_E$ does not change if $E$ varies inside the gap and so we can associate to each gap a projection in $\Ape$ and hence an element in $K_0(\Ape) $. 

The label associated to the gap is $\tp(P_E(H))$. This can be reinterpreted as the result of Connes' pairing between the cyclic cocycle defined by the trace $\tp$ and the $K_0$-class defined by $P_E(H)$. Indeed, the pairing of $\tp$ with a $K_0$-class $[p]_0-[s(p)]_0$, where $p$ is a projection in $M_n(\Ape^+)$, is 
$$\langle\tp, [p]_0-[s(p)]_0 \rangle = \tp \,\Tr(p-s(p))$$
where $\Tr$ is the matrix trace on $M_n(\C)$. Note that this
expression reduces to $\tp(p)$ if $p\in \Ape$ as then $s(p)=0$. In particular, if $H$ is bounded from below and resolvent affiliated to $\Ape$ then
$\tp(P_E(H)) = \langle\tp, [P_E(H)]_0 \rangle$. The map $\langle\tp,\cdot\rangle$ is also denoted $\tp_*$, as we did in the introduction.
\begin{definition}
The (continuous) gap labelling group associated to $\pe$ is
$$\gap := \langle\tp, K_0(\Ape)\rangle $$
where $\Ape$ is the continuous pattern algebra and $\tp$ the trace per unit volume. 
\end{definition}

While $\gap$ is a subgroup 
of $\R$ it should be kept in mind that its elements carry a dimension, namely $\mathrm{length}^{-d}$, as $\tp$ is the trace is per unit volume.

Of course, we don't need $K$-theory to define this gap-label or to calculate it. But the advantage of using $K$-theory is that we can compute the group $\langle\tp, K_0(\Ape)\rangle$ in many cases just from the data of $\pe$ and so restrict the possible set of gap-labels. The group $\gap$ plays the role of a selection rule: Whatever the potential $V$, as long as it is pattern equivariant all labels of the gaps of $H=-\Delta +V$ must belong to the gap labelling group.  This is the way Bellissard proved that for the tight binding Hamiltonian on the Thue-Morse chain, there can't be any gap with gap-label different from $\frac13\frac{m}{2^n}$, $m\in\NM, n\in\Z$ \cite{bellissard1990spectral}. 

Finally we should mention that the gap label $\tp(P_E(H))$ coincides with the integrated density of states of $H$ up to the gap. 
The integrated density of states in an energy interval is roughly speaking the number of eigenstates per unit volume whose energy belongs to the energy interval. To define it properly one has to first put the system in a finite box, count the number of eigenstates with energy below $E$ of the restriction of $H$ to the box, divide by the volume of the box, and then let the size of the box go to infinity. This can again be done with a van Hove sequence, as above for the trace per unit volume. There is, however, one issue to be controlled, namely that the limit exists and is not spoiled by spurious states caused by the boundary conditions which have to be imposed on the  restrictions of $H$ to the elements $\Lambda_n$ of the van Hove sequence. We refer to \cite{bellissard1992gap,kellendonk1995noncommutative,lenz2005ergodic} for this discussion about the so-called Shubin formula. 

Via this interpretation the gap-labelling has physical significance:  the integrated density of states of a Hamiltonian affiliated to $\Ape$ cannot take any value on a gap of its spectrum, but these values are constraint (topologically quantized) to the subgroup $\gap$ of $\R$. Note that $\gap$ is countable if $\Ape$ is separable. 

\subsection{Gap labelling for the tight binding approximation}
For tight binding approximations the above arguments are completely analogous when using the discrete pattern algebra $\Atb$, but for one detail: If, taking into account internal degress of freedom,  the tight binding Hamiltonian $H$ belongs to $M_N(\Atb)$ then the gap labels of $H$ will all lie in
$$\gapS \cap [0,N], \quad \gapS := \langle \ttb,K_0(\Atb))\rangle,$$
and $\langle \ttb,K_0(\Atb))\rangle = \frac1{\mathrm{dens}(S)} \langle \tp,K_0(\Ape)\rangle$.
This simply follows from (\ref{eq-tr}) and the fact that any spectral projection of $H$ lies below the unit $1_N\in M_N(\Atb)$.

The gap-labelling conjecture says that,
if $\pe$ has finite local complexity then $ \gapLS$ is the frequency module of $\pe$, that is, the subgroup of $\R$ which is generated by the frequencies $ \nu(\Pp)$ of the patch classes $\Pp$ of $\pe$.
The conjecture has been proven for $d\leq 3$ \cite{bellissard2001gap}, see also \cite{kaminker2003proof,bellissard2006spaces,benameur2007index}.

\section{Bragg spectrum}\ref{sec-Bragg}
In an X-ray picture of the solid, the bright spots are referred to as Bragg peaks and the 
Bragg spectrum is the set of positions of these peaks. They depend on the precise conditions of the X-ray experiment. Mathematically they are determined by the diffraction measure associated to the charge distribution of the solid \cite{baake2013aperiodic}. A Bragg peak can be identified with a point mass component of the diffraction measure, and its intensity with the value of the measure at that point. So the mathematical definition of the Bragg spectrum is as the support of the pure point part of the diffraction measure.
 
Let $\omega$ be the electronic density at which the $X$-rays are diffracted. We suppose that $\omega$ is a continuous function and argued that it should be pattern equivariant, $\omega \in  C_\pe(\R^d)$. One may take for example $\omega = \sum_{a\in A} \rho_a \ast \delta_{\pe_a}$ where $\rho_a$ is the density contribution from atom of type $a$.  We use the van Hove sequence $(\Lambda_n)_n$ which has been used above for the definition of the trace per unit volume $\mean$ to define the autocorrelation distribution $\gamma_\omega$ 
$$\mathcal \gamma_\omega = \lim_{n\to \infty}\frac1{|\Lambda_n|} \omega_{\left|_{\Lambda_n}\right.}\ast \tilde\omega_{\left|_{\Lambda_n}\right.} $$
where $\tilde w (x) = \overline{w (-x)}$. Existence of the limit follows from Prop.~1.4 of \cite{lenz2020pure}, as for compactly supported continuous $\varphi$ the function $\varphi \ast \omega$ is pattern equivariant.

The Fourier transform $\hat\gamma_\omega$ of the autocorrelation, defined in the distributional sense,
turns out to be a translation bounded positive measure \cite{baake2013aperiodic,lenz2020pure}. This is the diffraction measure. As a measure on $\widehat{\R^d}$ (the Pontrayagin dual of $\R^d$)  it can be decomposed into its pure point, singular continuous and absolutely continuous part w.r.t.\ the Lebegue measure. By definition, the Bragg spectrum is thus the set $\{\chi\in\hat \R^d | \hat\gamma_\omega(\{\chi\})\neq 0\}$. As is customary in the physical literature (and also in \cite{baake2013aperiodic}) we use the isomomorphism ${\R^d}^*\ni k \mapsto \chi_k \in \widehat{\R^d}$
$$\chi_k(x) = e^{i\langle k,x\rangle}$$
 to identify the Pontrayagin dual with $k$-space, the vector space dual of $\R^d$, and write the Bragg spectrum as
$$\mathcal B(\omega) = \{k\in{\R^d}^* | \hat\gamma_\omega(\{\chi_k\})\neq 0\}.$$ 

To some degree, $\mathcal B(\omega)$ is robust to change of the local charge distributions $\rho_a$ and depends only on $\pe$.  As long as Bragg peaks do not extinct, $\mathcal B$ remains the same.  This is a simple consequence of the fact that $\rho_a$ contributes only through multiplication with $\hat \rho_a$ to the pure point part of $\hat \gamma_\omega$. This has some simililarity with the gap labelling, where a gap label does not change if the gap does not close.

Interesting questions to ask are, how large is $\mathcal B(\omega)$, and how large can it possibly be? While the first question can be answered only through analytical calculations (the actual calculation that the intensity of a potential peak is non-zero), the second question can be answered by more global techniques, here coming from dynamical systems theory. 

\subsection{The dynamical spectrum relevant for diffraction} 
It is one of the corner stones of diffraction theory that the Bragg spectrum is related to the eigenvalues of the dynamical spectrum of $\pe$. This connection, which is also referred to as Dworkin's argument \cite{dworkin1993spectral}, has been developed by \cite{hof1995diffraction} and then many other people, 
and found a very general formulation for ergodic translation bounded measure dynamical systems on locally compact abelian groups in \cite{baake2004dynamical}. We present it here in a form adapted to our framework which is dual to the one in which the translation action on the pattern space $\Omega_\pe$ is considered.

We consider the GNS representation of $C_\pe(\R^d)$ w.r.t.\ the state $\mean$. Its Hilbert space is the completion of $C_\pe(\R^d)$ under the seminorm defined by
$$\| f \|^2 := \mean(|f|^2)$$
and we will denote it by $L^2_\pe(\R^d,\mean)$. This completion involves a quotient if $\mean$ is not faithful. We provide a criterion for that. Denote by $[f]_\mean$ the class of $f$ in  $L^2_\pe(\R^d,\mean)$.
\begin{lemma} 
Suppose that $\pe$ is relatively dense.
$\mean$ is faithful on $C_\pe(\R^d)$ if and only if all patch classes have non-zero frequency.
In that case we may view $C_\pe(\R^d)$ as a subspace of $L^2_\pe(\R^d,\mean)$.
\end{lemma}
\begin{proof} 
As $\mean$ is continuous on $C_\pe(\R^d)$ (w.r.t.\ the sup-norm topology) it suffices to show that it is faithful on the dense subalgebra of strongly pattern equivariant functions. Let $f$ be a positive function which is strongly pattern equivariant with radius $R$. 
By relative denseness there is $M$ be such that $\pe+B_M(0)$ covers $\R^d$. Let $\{\mathcal P_i\}_{i\in I}$ be the (at most countable) set of $R+M$-patch classes which have a representative with center in $\pe$. Then $\int_{B_{M}(p)} f(x) dx $ is the same for all $p\in \pe_{\mathcal P_i}$. 
Therefore, if we choose for each $i$ a $p_i\in \pe_{\mathcal P_i}$ we get
$$\mean(f) = \sum_{i\in I} \nu(\mathcal P_i) \frac1{|B_{M}(p_i)|} \int_{B_{M}(p_i)} f(x) dx $$
where $\nu(\mathcal P_i)$ is the frequency of $\mathcal P_i$. Thus $\mean(f)=0$ implies $f(x)=0$ for all $x\in B_M(p)$, $p\in \pe_{\mathcal P_i}$ for which  $\nu(\mathcal P_i)\neq 0$. If all frequencies are non-zero then $f=0$, as $\pe+B_M(0)$ covers $\R^d$.

As for the converse, if $\nu(\mathcal P)=0$ then $f = \kappa\ast \delta_{\pe_P}$ is a strictly positive function for which $\mean(f)= 0$, where $\kappa$ is a positive function with compact support.
\end{proof}
$L^2_\pe(\R^d,\mean)$ does not only carry the action of $C_\pe(\R^d)$ induced by left multiplication, but also a unitary representation of $\R^d$ by translation, $$U_a[f]_m = [f(\cdot + a)]_m.$$ The spectrum of this representation is the dynamical spectrum relevant for diffraction.

The $C^*$-dynamical system with its state $(C_\pe(\R^d),\R^d,\alpha,\mean)$ is dual in the sense of Gelfand to a measure dynamical system $(\Omega_\pe,\R^d,\alpha,\mu)$ where, as already mentionned, $\Omega_\pe$ is the Gelfand spectrum of the algebra and $\mu$ the probability measure on $\Omega_\pe$ corresponding via Riesz' theorem to the normalised state $\mean$. Under this duality, the Hilbert space $L^2_\pe(\R^d,\mean)$ corresponds to $L^2(\Omega_\pe,\mu)$, the completion of $C(\Omega_\pe)$ w.r.t.\ the semi-norm given by $\|f\|^2=\int_{\Omega_\pe} |f|^2 d\mu$.

\subsection{Diffraction to dynamics map}
The map relating diffraction to dynamics combines the Fourier-transform 
$$\hat \varphi(\chi) = \int_{\R^d} \overline{\chi(x)} \varphi(x) dx $$
with convolution by the density $\omega$. We denote by $\mathcal S(\hat\R^d)$ the space of Schwarz-functions on $\hat\R^d$.
\begin{theorem}[\cite{moody2008recent,baake2004dynamical}]\label{thm-d2d}
Let $\omega\in C_\pe(\R^d)$ and $\gamma_\omega$ be its autocorrelation (defined via the same van Hove sequence as $\mean$) 
The map $\Theta_\rho : \mathcal S(\hat\R^d) \to L^2_\pe(\R^d,\mean)$,
$$\Theta_\omega (\hat \varphi)  = \varphi\ast \omega $$
extends by continuity to an isometry 
$$\Theta_\omega : L^2(\hat\R^d,\hat\gamma_\omega) \to L^2_\pe(\R^d,\mean)$$
which intertwines the $\R^d$-action on $L^2(\hat\R^d,\hat\gamma_\omega)$ given by $\hat U_a \hat\varphi(\chi) = \chi(a) \hat\varphi(\chi)$ with the translation action $U_a$ on $L^2_\pe(\R^d,\mean)$. 
\end{theorem}
\begin{proof} We adapt the arguments of \cite{baake2004dynamical} to our framework. Let $ \varphi, \psi \in \mathcal S(\hat\R^d)$.
We have
$$\widehat{\overline{\hat \varphi} \hat \psi}(x) = \tilde\phi\ast \psi (-x)$$ and therefore
$$\int \overline{\hat \varphi}(\chi) \hat \psi(\chi) d\hat\gamma_\omega(\chi) = \int \tilde\phi\ast \psi (-x)d\gamma_\omega(x) = \tilde\phi\ast \psi \ast\gamma_\omega(0).$$
Furthermore
$$\langle \Theta_\omega(\hat \varphi),\Theta_\omega(\hat \psi)\rangle = \mean(\overline{\varphi\ast\omega}\cdot 
\psi \ast \omega) = \tilde \varphi \ast \psi \ast \gamma_\omega (0)$$
where the second equality follows from the fact that $\varphi$ and $\psi$ decay fast, so that we can replace $\left.\varphi\ast\omega\right|_{\Lambda_n}$ by  $\left.(\varphi\ast\omega)\right|_{\Lambda_n}$ the difference vanishing in the limit $n\to \infty$ due to the van Hove property, see Prop.~1.4 of \cite{lenz2020pure} for more details.

$\R^d$-equivariance follows from the fact that the Fourier transform intertwines the two actions. 
\end{proof}
\begin{definition} A vector $k\in {\R^d}^*$ is called an eigenvalue of the $\R^d$-action on $L^2_\pe(\R^d,\mean)$ if there is a non-zero class $[f]_\mean\in L^2_\pe(\R^d,\mean)$ such that, for all $a\in\R^d$,
$$ U_a[f]_\mean = e^{i\langle k,a\rangle} [f]_\mean $$
We denote the eigenvalues by $\mathcal E_\pe$. 

A vector $k\in {\R^d}^*$ is called a topological eigenvalue of the $\R^d$-action on $C_\pe(\R^d)$ if there is a non-zero  function $f \in C_\pe(\R^d)$ such that
$$ U_af = e^{i\langle k,a\rangle} f $$ 
for all $a\in\R^d$.
We denote the topological eigenvalues by $\Et$. 
\end{definition} 
Eigenvalues of the $\R^d$-action on $L^2_\pe(\R^d,\mean)$ are also referred to as measurable eigenvalues of the dynamical system associated to the solid. Note that topological eigenvalues do not depend on the choice of $\mean$. Indeed, they are characterised by the fact that the plane wave $x\mapsto e^{i\langle k,x\rangle}$ is pattern equivariant:
\begin{lemma}
$k\in {\R^d}^*$ is a topological eigenvalue iff the map $x\mapsto e^{i\langle k,x\rangle}$ is pattern equivariant.
\end{lemma}
\begin{proof} Let $k$ be a topological eigenvalue. Then there exists $f\in C_\pe(\R^d)$ and $x_0\in\R^d$ such that $f(x_0)\neq 0$ and $U_af = e^{ika} f$ for all $a\in\R^d$. Then $e^{i\langle k,x\rangle} = f(x_0)^{-1}f(x+x_0)$ which clearly belongs to $C_\pe(\R^d)$. As for the converse, a pattern equivariant plane wave is a continuous eigenfunction to its wave vector.
\end{proof}
The topological eigenvalues of $\pe$ characterise $\pe$ to quite some extend \cite{aujogue2015equicontinuous}, for instance, if $\pe$ has  $d$ linearly independent topological eigenvalues then it satisfies the Meyer property up to  topological conjugacy \cite{kellendonk2014meyer}. 

Theorem~\ref{thm-d2d} provides us with a map from the Bragg spectrum $\Bb(\omega)$ to $\mathcal E_\pe$. This simply follows from the fact that if $\hat\gamma(\{\chi_k\})\neq 0$ then the $L^2$-class of the indicator function $1_{\{\chi_k\}}$ is non-zero, hence also 
 $\Theta_\omega([1_{\{\chi_k\}}])$ and, by equivariance
$$U_a \Theta_\omega([1_{\{\chi_k\}}]) = \Theta_\omega(\hat U_a [1_{\{\chi_k\}}]) = e^{i\langle k,a\rangle} \Theta_\omega([1_{\{\chi_k\}}]).$$
So $\Bb(\omega)$ is a subset of $\mathcal E_\pe$. It might not be all of $\mathcal E_\pe$ but the question of when this is the case depends on the specific choice of the density $\omega$ and is not of topological nature.

$\mathcal E_\pe$ and $\Et$ are groups, as the product of two eigenfunctions yields an eigenfunction to the sum of their corresponding eigenvalues. Furthermore, $\Et$ is a subgroup of $\mathcal E_\pe$, for
$k \mapsto  \left[e^{i\langle k,x\rangle}\right]_\mean$ is injective, as $\mean(|e^{i\langle k,x\rangle}-e^{i\langle k',x\rangle}|^2)= 4\mean(\sin^2((k-k')x))>0$ if $k\neq k'$. For many long range ordered structures, like quasiperiodic ones or those defined by substitution rules, $\Et$ coincides with $\mathcal E_\pe$. Otherwise, if $\mean$ is faithful one has the following criterion for an eigenvalue to be topological.
\begin{lemma} Suppose that $\mean$ is faithful.
If $\Theta_\omega([1_{\{k\}}])$ has a continuous representative then $k$ is a topological eigenvalue and the representative is given by a pattern equivariant plane wave 
$x\mapsto \hat\gamma(\{\chi_k\}) e^{i\langle k,x-x_0\rangle}$
with some $x_0\in \R^d$. 
\end{lemma}
\begin{proof}
Let $f$ be continuous and $\Theta_\omega(1_{\{\chi_k\}}) =[f]_\mean$. Then, for all $x$, 
$0=(U_x-\chi_k(x)) [f]_m = [U_x f - \chi_k(x) f]_m$. Hence $\mean(|U_x f - \chi_k(x) f|)=0$. As $\mean$ is faithful, 
$U_x f - \chi_k(x) f=0$. Hence $k$ is a topological eigenvalue. Now $f(y+x) = e^{i\langle k,x\rangle} f(y)$ shows that $|f(x)|$ is constant. Hence $|f(0)|^2 = \mean(|f|^2) = \|1_{\{\chi_k\}}\|^2$, the latter as $\Theta_\omega$ is a isometry. 
Finally $\|1_{\{\chi_k\}}\|^2= \hat\gamma(\{\chi_k\})^2 $. 
\end{proof}
We point out that under the criterion of the last lemma, the amplitude $\hat\gamma(\{\chi_k\})$ of the Bragg peak at $k$ can be computed using the Bombieri-Taylor formula \cite{lenz2009continuity}. 
\section{From topological eigenvalues to Chern numbers}\label{sec-main}
We now relate topological eigenvalues to the $K$-group of the pattern algebra $\Ape$ and its image under Connes' pairing with certain naturally defined cyclic cocycles. 
This relation is based on the fact that $\chi_k(x) = e^{i\langle k,x\rangle}$ is a pattern equivariant function and hence a unitary of $C_\pe(\R^d)$. We therefore obtain 
a group homomorphism
$\Et \stackrel{\varphi}\to K_1(C_\pe(\R^d))$, 
$$\varphi(k) = [\chi_k]_1$$
which will be combined with the Connes Thom isomorphism to obtain a map into  the observable algebra $\Ape$. Our aim is to fill in the details and prove diagram (\ref{eq-main-diagram}) which we will do by decomposing it in two commuting diagrams. But first we need some more background material.

\subsection{$K$-theory of crossed products with $\R^d$}
One of the fundamental results in $K$-theory which is of importance in this work is Connes' isomorphism \cite{connes1981analogue,blackadar1998k}. 
Given a $C^*$-algebra $A$ with an action $\alpha$ of $\R$ by automorphisms $\alpha^t$, $t\in \R$ such that for all $a\in A$ the map $\R\ni t\to \alpha^t(a)$ is continuous, the Connes Thom isomorphism relates the $K$-theory of $A$ to the $K$-theory of the crossed product $A\rtimes_\alpha\R$,
$$ K_i(A) \stackrel{\phi_\alpha}{\to} K_{i+1}(A\rtimes_\alpha \R) .$$
It is uniquely determined by its functorial properties and the requirement that $\phi_{\id}$ maps the generator of $K_0(\C)$ to the generator of $K_1(\C\rtimes_{\id}\R)$ (this involves a choice of a orientation). 
It can also be obtained as (minus) the boundary map of a smooth Toeplitz extension, or as (minus) the inverse of the boundary map of the Wiener Hopf extension, see \cite{schulz2022harmonic}. 

If $A$ carries a (strongly continuous) action $\alpha$ of $\R^d$ then the crossed product $A\rtimes_\alpha\R^d$ is isomorphic to an iterated crossed product with $\R$. Indeed, choosing a basis $\Bb=(e_1,\cdots,e_d)$ of $\R^d$ we obtain a collection of $d$ commuting $\R$-actions $\alpha_1,\cdots,\alpha_d$, where $\alpha_i^t := \alpha^{t e_i}$. We assume that the parallel-epiped defined by $(e_1,\cdots,e_d)$ has Lebesgue measure $1$. 
Then 
$\Psi_\Bb:A\rtimes_\alpha\R^d \to A\rtimes_{\alpha_1}\R \cdots\rtimes_{\alpha_d}\R$
defined on $F\in C_c(\R^d,A)$ by
$$\Psi_\Bb(F)(t_d)\cdots(t_1) := F(\textstyle \sum_i t_i e_i)$$
extends by continuity to an isomorphism of $C^*$-algebras.
Note that, if $d=1$ so that $\Bb=(e_1)$ then ${\Psi_{(-e_1)}}_* = -  {\Psi_{(e_1)}}_*$ as the base change $e_1\mapsto -e_1$ changes the orientation. Likewise 
exchanging the order of two actions changes the orientation and so we obtain
\begin{equation}\label{eq-B-Thom}
{\Psi_\Bb}_* = \det(\Bb,\Bb'){\Psi_{\Bb'}}_*
\end{equation}
where $\det(\Bb,\Bb')$ is the determinant of the change of basis which must be $\pm 1$ as the parallel-epiped defined by the basis must have Lebesgue measure $1$. It follows that  ${\Psi_\Bb}_*$ depends only on the orientation $e:=e_1\wedge\cdots\wedge e_d$ and we can denote it as ${\Psi_e}_*$.

Accordingly, we may iterate the Connes Thom isomorphism and define
$$\phi_\alpha^e := {\Psi_e}_*^{-1}\phi_{\alpha_d}\cdots \phi_{\alpha_1} $$
to obtain an isomorphism 
$\phi_\alpha^e : K_i(A) \to K_{i+d}(A\rtimes_{\alpha}\R^d)$.

\subsection{Chern cocycles}
We recall the definition of certain cocycles which appear in the context of solid state theory and their pairing with $K$-theory. A good summary of the material can be found in the book \cite{schulz2022harmonic} which we follow in all its conventions.

As above we consider a $C^*$-algebra $A$ with an action $\alpha$ of $\R^d$ which we decompose into a collection of $d$ commuting $\R$-actions $\alpha_1,\cdots,\alpha_d$. These actions give rise to densely defined derivations 
$$\nabla_i(a) :=-\frac1{2\pi}\left.  \partial_t{\alpha^t_{i}(a)} \right|_{t=0}. $$
Let $\Tt$ be a densely defined l.s.c.\ trace which is invariant under the actions, $\Tt\circ\alpha_i^t =\Tt$. Let $A_{\Tt,\alpha}$ be the sub-algebra of elements of $A$ which are smooth w.r.t.\ the derivations $\nabla_i$.
The Chern cocycle defined by the actions $\alpha_1,\cdots,\alpha_d$ and the trace $\Tt$ 
is the multilinear map $\mathrm{Ch}_{\Tt,\alpha_1,\cdots,\alpha_d}:{A_{\Tt,\alpha}}^{d+1}\to \C$ 
given by
$$\mathrm{Ch}_{\Tt,\alpha_1,\cdots,\alpha_d}(a_0,\cdots,a_d) = c_d \sum_{\sigma\in S_d} \mathrm{sgn}(\sigma) \Tt(a_0 \nabla_{\sigma(1)} (a_1) \cdots \nabla_{\sigma(d)} (a_d))$$
where
$$c_d =\left\{ \begin{array}{ll}
\frac{(2\pi i)^k}{k!},& \mathrm{if }\: d=2k \\ 
\frac{i(\pi i)^k}{(2k+1)!!},& \mathrm{if }\:  d=2k+1
\end{array}\right.$$
Note that $\mathrm{Ch}_{\Tt,\alpha_1,\cdots,\alpha_d}$ depends on the choice of base. 
It extends to $A_{\Tt,\alpha}$-valued matrices $M_n(A_{\Tt,\alpha})\cong A_{\Tt,\alpha}\otimes M_n(\C)$ by extending the actions (and hence derivations) entrywise $\alpha_i\otimes \id$ and extending $\Tt$ using the matrix trace $\Tt\otimes\Tr$. 
The Connes pairing between $\mathrm{Ch}_{\Tt,\alpha_1,\cdots,\alpha_d}$ and an element $[x]_i$ of 
$K_i(A_{\Tt,\alpha})$ is defined as follows. 
If $i$ is even, then $[x]_i$ has the form $[p]_0-[s(p)]_0$ where $p$ is a projection in $M_n({A_{\Tt,\alpha}}^+)$ and
$$\langle \mathrm{Ch}_{\Tt,\alpha_1,\cdots,\alpha_d} , [p]_0-[s(p)]_0\rangle = 
\mathrm{Ch}_{\Tt,\alpha_1,\cdots,\alpha_d}(p-s(p),\cdots, p- s(p))$$
If $i$ is odd, then $K_i(A)$ is given by elements of the form $[u]_1$ where 
$u$ is a unitary in $M_n({A_{\Tt,\alpha}}^+)$ and
$$\langle \mathrm{Ch}_{\Tt,\alpha_1,\cdots,\alpha_d} , [u]_1\rangle = 
\mathrm{ch}^{(n)}_{\Tt,\alpha}(u^*,u,\cdots, u^*,u)$$

There is a isomorphism in cyclic cohomology \cite{elliott1988cyclic} which is dual to the Connes Thom isomorphism $\phi_{\alpha_i}:K_n(A) \to K_{n+1}(A\rtimes_\alpha \R)$.  This involves the dual action $\hat\alpha_i$ and the dual trace $\hat\Tt$ on the crossed product. On the dense sub-algebra $C_c(\R,A)$ (with $\alpha_i$-twisted convolution product) they are defined by 
$$\hat\alpha_i^s(f) = \chi_{se^i} f,\quad \hat\Tt(f) = \Tt(f(0)).$$
Here $e^i$ is the corresponding element of the dual basis.
In particular, the corresponding dual derivation is given by 
$$\hat \nabla_i f = -\frac1{2\pi}\left. \partial_s\chi_{se^i} \right|_{s=0} f .$$
The original action $\alpha_i$ extends to the crossed product via $\alpha_i^s(f)(t)=\alpha_i^s(f(t))$, and this extension commutes with $\hat\alpha_i$. 

Given $1\leq j,n\leq d$ we compare now three Chern cocyles. The first is 
$ \mathrm{Ch}_{\Tt,\alpha_1,\cdots,\alpha_{n}}$, the Chern cocycle on 
$A_{\Tt,\alpha_1,\cdots,\alpha_{n}}$ defined by $n$ of the $d$ commuting actions. The other two Chern cocycles are defined on the smooth subalgebra $(A\rtimes_{\alpha_n}\R)_{\hat\Tt,\alpha_1,\cdots,\alpha_{n-1},\hat\alpha_n}$ of the crossed product. Namely we consider $\mathrm{Ch}_{\hat\Tt,\alpha_{1},\cdots,\alpha_n,\hat\alpha_j}$ the Chern cocycle 
 defined by the $n+1$ commuting actions $\alpha_{1},\cdots,\alpha_n,\hat\alpha_j$ of which the last is a dual action and where we have extended the other actions to the crossed product $A\rtimes_{\alpha_n}\R$, and finally, if $j=n$, $\mathrm{Ch}_{\hat\Tt,\alpha_{1},\cdots,\alpha_{n-1}}$, the Chern cocycle defined by the first $n-1$ of the $n$ commuting actions extended to the crossed product. These Chern cocycles are related via the Connes Thom isomorphism through the following theorem.
\begin{theorem} \label{thm-SS}
Let $A$ be a $C^*$-algebra with $d$ commuting $\R$ actions 
$\alpha_1,\cdots,\alpha_d$ and $\alpha_i$-invariant faithful densely defined l.s.c.\ trace $\Tt$. Let $1\leq j,n\leq d$. 
Let $\hat\alpha_j$ be the dual action on
the crossed product $A\rtimes_{\alpha_j}\R$ and  $\hat\Tt$ the dual trace. Then
$$\langle \mathrm{Ch}_{\Tt,\alpha_{1},\cdots,\alpha_n},[x]_i\rangle = (-1)^{n}\langle \mathrm{Ch}_{\hat\Tt,\alpha_{1},\cdots,\alpha_n,\hat\alpha_j},\phi_{\alpha_j}[x]_i\rangle$$ 
where $[x]_i\in K_{i}(A)$.
If furthermore $j=n$ then also
$$\langle \mathrm{Ch}_{\Tt,\alpha_{1},\cdots,\alpha_n},[x]_i\rangle = (-1)^{n+1}\langle \mathrm{Ch}_{\hat\Tt,\alpha_{1},\cdots,\alpha_{n-1}},\phi_{\alpha_n}[x]_i\rangle$$
\end{theorem}
\begin{proof} The two equations are Theorem~4.5.2 and Theorem~4.5.3 of 
\cite{schulz2022harmonic} together with the identification of the Connes Thom isomorphism with (minus) the inverse of the connecting map of the Wiener Hopf extension for the first, and with (minus) the connecting map of the smooth Toeplitz  extension. 
\end{proof}
We aim to write the above in a way which is invariant under a change of basis of $\R^d$ which preserves the orientation. For that we collect the derivations into the differential $\delta:A_{\Tt,\alpha}\otimes  \Lambda{\R^*}^d\to A_{\Tt,\alpha}\otimes  \Lambda{\R^*}^d$
$$ \delta(a\otimes\lambda) = \sum_{i=1}^d \nabla_i (a)\otimes e^i\wedge\lambda.$$
It is invariant under a base change. 
The dual action is given by $\hat \alpha^k f = \chi_k f$ with $\chi_k(x) = e^{i\langle k,x \rangle}$ and hence its components have to be expressed in the dual basis, $\hat \alpha_i^t = \hat\alpha^{te^i}$. The corresponding differential is 
$\hat\delta:(A\rtimes_\alpha\R^d)_{\hat\Tt,\hat\alpha}\otimes  \Lambda{\R^*}^d\to (A\rtimes_\alpha\R^d)_{\hat\Tt,\hat\alpha}\otimes  \Lambda{\R}^d$
$$ \hat\delta (f\otimes\lambda) = \sum_{i=1}^d \hat \nabla_i (f)\otimes e_i\wedge\lambda$$
which we can also write as
\begin{equation}\label{eq-deltahat}
 \hat\delta (f\otimes\lambda)(x) = \imath f(x)\otimes x\wedge\lambda
\end{equation}
also invariant under a base change.
Define, for $n\leq d$, 
$\mathrm{ch}^{(n)}_{\Tt,\alpha} : {A_{\Tt,\alpha}}^{n+1}\to \Lambda{\R^*}^d$ ,
$$\mathrm{ch}^{(n)}_{\Tt,\alpha}(a_0,\cdots,a_n) =  \Tt\otimes\id(a_0 \delta a_1 \cdots \delta a_n)$$
and
$\mathrm{ch}^{(n)}_{\hat\Tt,\hat\alpha} : (A\rtimes_\alpha\R^d)_{\hat\Tt,\hat\alpha}^{n+1}\to \Lambda{\R}^d$ ,
$$\mathrm{ch}^{(n)}_{\hat\Tt,\hat\alpha}(f_0,\cdots,f_n) =  \hat\Tt\otimes\id (f_0 \hat\delta f_1 \cdots \hat\delta f_n).$$
These quantities are invariant under a change of basis.

The dual pairing between $\R^d$ and ${\R^d}^*$ extends to the exterior modules and we denote this extension by $\langle \cdot,\cdot\rangle_\Lambda$.
For $v\in \Lambda^n {\R^d}^*$ let $D^e(v)$ be the unique element of 
$\Lambda^{d-n} \R^d$ which satisfies
$$\langle D^e(v),w\rangle_\Lambda = \langle  v\wedge w , e \rangle_\Lambda$$ 
for all $w\in \Lambda^{d-n} {\R^d}^*$. Here, we recall, $e=e_1\wedge\cdots\wedge e_d$ and we required that the corresponding parallel epiped has Lebesgue measure $1$. 
This defines an isomorphism $D^e:\Lambda {\R^d}^* \to \Lambda \R^d$, often referred to as Poincar\'e isomorphism, which depends only on the orientation $e$ of the basis. 
With these definitions we conclude from Theorem~\ref{thm-SS}
 \begin{theorem}\label{thm-diagram-2}
 Let $A$ be a $C^*$-algebra with an action $\alpha$ of $\R^d$ and $\alpha$-invariant faithful densely defined l.s.c.\ trace $\Tt$. Let $\hat\alpha$ be the dual action and $\hat\Tt$ the dual trace on the crossed product 
$A\rtimes_{\alpha}\R^d$. 
Let $n\leq d$. For any choice of orientation $e$ on $\R^d$ the following diagram is commutative
 $$\begin{array}{ccc}
 K_{n}(A) & \stackrel{\phi^e_\alpha} \to & K_{n+d}(A\rtimes_\alpha\R^d)  \\
  \downarrow  \langle\ch^{(n)}_{\Tt,\alpha}  ,\cdot\rangle    &  & \downarrow 
 \epsilon_{d-n}\langle \ch^{(d-n)}_{\hat\Tt,\hat\alpha} ,\cdot\rangle \\       
\Lambda^n{\R^d}^*  & \stackrel{D^e}\to &  \Lambda^{d-n}\R^d
 \end{array}$$
 where 
$\epsilon_k = (-1)^{N_k}$ with $N_k$ the number of odd natural numbers strictly below $k$.
 \end{theorem}
\begin{proof} Let $\Bb=e_1,\cdots,e_d$ be a base of $\R^d$ s.th.\ $e=e_1\wedge \cdots\wedge e_d$.
We have to show that, for all $w\in\Lambda^{d-n}{\R^d}^*$, 
$$\langle \langle \ch^{(n)}_{\Tt,\alpha} ,[x] \rangle \wedge w , 
e_1\wedge\cdots\wedge e_d \rangle_\Lambda =  \epsilon_{d-n}
\langle  w , \langle \ch^{(d-n)}_{\hat\Tt,\hat\alpha} ,{\Psi_e}_*^{-1}\phi_{\alpha_d}\cdots\phi_{\alpha_1}[x] \rangle 
 \rangle_\Lambda. $$
By the second formula of Theorem~\ref{thm-SS}
$$\langle \mathrm{Ch}_{\Tt,\alpha_{1},\cdots,\alpha_{n}},[x]\rangle =
(-1)^{n-1} \langle \mathrm{Ch}_{\Tt,\alpha_{2},\cdots,\alpha_{n},\alpha_1},[x]\rangle = \langle \mathrm{Ch}_{\hat\Tt,\alpha_{2},\cdots,\alpha_{n}},\phi_{\alpha_1}[x]\rangle$$
It follows iteratively that 
$$\langle \mathrm{Ch}_{\Tt,\alpha_{1},\cdots,\alpha_{n}},[x]\rangle =
\langle \mathrm{Ch}_{\hat\Tt},\phi_{\alpha_n}\cdots \phi_{\alpha_1}[x]\rangle$$
Now the first formula yields
$$\langle \mathrm{Ch}_{\hat\Tt},\phi_{\alpha_n}\cdots \phi_{\alpha_1}[x]\rangle = (-1)^n \langle \mathrm{Ch}_{\hat\Tt,\hat\alpha_{n+1}},\phi_{\alpha_{n+1}}\cdots \phi_{\alpha_1}[x]\rangle
$$
and we obtain iteratively
$$\langle \mathrm{Ch}_{\Tt,\alpha_{1},\cdots,\alpha_{n}},[x]_i\rangle = \epsilon_{d-n} \langle \mathrm{Ch}_{\hat\Tt,\hat \alpha_{n+1},\cdots,\hat\alpha_{d}},\phi_{\alpha_d}\cdots\phi_{\alpha_1}[x]_i\rangle.$$ 
This proves the formula for $w= e^{n+1}\wedge \cdots \wedge e^d$.

The other choices for $w$ can be related to the above by exchanging iteratively two basis elements. 
This introduces various signs which match to yield the result. We show this in the simpler case $n=1$, 
$d=2$, the general case following similarily.
\begin{eqnarray*}
\langle \langle \ch^{(1)}_{\Tt,\alpha} ,[x] \rangle \wedge e^1 , 
e_1\wedge e_2 \rangle_\Lambda & = &
-\langle\Ch_{\Tt,\alpha_2},[x]\rangle\\
 & = & 
 +\langle\Ch_{\hat\Tt,\alpha_2,\hat\alpha_1,\hat\alpha_2},\phi_2\phi_1[x]\rangle \\
& = & -\langle\Ch_{\hat\Tt,\hat\alpha_1,\alpha_2,\hat\alpha_2},\phi_2\phi_1[x]\rangle \\
& = & \epsilon_1\langle e^1,\langle\Ch_{\hat\Tt,\hat\alpha_1},\phi_2\phi_1[x]\rangle\rangle_\Lambda
\end{eqnarray*}
Above, the second and fourth equality follow from Theorem~\ref{thm-SS} and the third by antisymmetry of the Chern-cocycle w.r.t.\ an exchange of the order of the actions. 
\end{proof}

The constants $c_d$ have been chosen such that, besides the sign, no other numerical factor arises in Theorem~\ref{thm-SS} and $\langle \mathrm{ch}^{(0)}_{\id},K_0(\C)\rangle = \Z $. This implies that
$$\langle \mathrm{ch}^{(n)}_{\hat\Tt,\hat\id},K_n(\C\rtimes_{\id}\R^n)\rangle = \Z e_1\wedge\cdots \wedge e_n$$ 
and means that $K_n(\C\rtimes_{\id}\R^n)$ (which is isomorphic to $\Z$) contains an element $[x^{(n)}]_n$ which is a generator and satisfies $\langle \mathrm{ch}^{(n)}_{\hat\Tt,\hat\id},[x^{(n)}]_n\rangle = e_1\wedge\cdots \wedge e_n$. We refer to that element as the Bott element in $K_n(\C\rtimes_{\id}\R^n)$.

Fourier transformation provides an isomorphism between 
$\C\rtimes_{\id}\R^n$ and $C_0(\R^n)$ and under this isomorphism 
$\hat \Tt$ becomes $\int$, integration w.r.t.\ Lebesgue measure, and $\hat\id$ the translation action which we denote $\alpha$. Let $[\tilde b^{(n)}]_n - [s(\tilde b^{(n)})]_n$ be the 
image of the element $[x^{(n)}]_n$ under the induced morphism on $K$-theory. 
By homotopy invariance we may assume that the support of its representative $\tilde b^{(n)}-s(\tilde b^{(n)})$ is contained in the interior of the unit cube (w.r.t.\ the basis) $I\subset\R^n$. Let $b^{(n)}:\R^n\to M_m(\C)$ be given by
$$b^{(n)}(t) =\tilde  b^{(n)}(t-z)$$
where $z\in \Z^n$ (the lattice spanned by the basis) is such that $t-z\in I$. Then 
$b^{(n)}(t)$ is periodic, i.e.\
$b^{(n)}\in M_m(C_{\Z^n}(\R^n))$ and 
\begin{equation}\label{eq-norm}
1=\langle \ch^{(n)}_{\hat\Tt,\hat\id},[x^{(n)}]_n\rangle = 
\langle \ch^{(n)}_{\int ,\alpha},[ \tilde b^{(n)}]_n\rangle =
\langle \ch^{(n)}_{\mean,\alpha},[ b^{(n)}]_n\rangle
\end{equation}
where for the last equation we need that the volume of $I$ is $1$.

\subsection{The first commuting diagram}
There is a functorial way to extend the inclusion $\imath :\Et \to {\R^d}^*$ to the exterior moduls:
$$\Lambda^n\imath : \Lambda^n\Et \to \Lambda^n {\R^d}^* .$$
Indeed,
$$\Lambda^n\imath(k_1\wedge \cdots \wedge k_n) =
\sum_{i_1<i_2\cdots <i_n} \det(k_1,\cdots,k_n)_{i_1,\cdots,i_n} e^{i_1}\wedge\cdots\wedge e^{i_n}$$
where  $\det(k_1,\cdots,k_n)_{i_1,\cdots,i_n}$ is the determinant of the following $n\times n$ matrix: 
Each $k\in{\R^d}^*$ defines a column by taking its coefficients in the basis $e^1,\cdots,e^d$, hence $\begin{pmatrix} k_1 & \cdots & k_n \end{pmatrix}$ is a $d\times n$ matrix. We keep of this matrix  only the rows $i_1,\cdots,i_n$ to obtain the matrix of which we take the determinant.
\bigskip

We now use the above elements $b^{(n)}$ to extend the map $\varphi:\Et \to K_1(C_\pe(\R^d))$ to the exterior module.
Note that, if $f:\R\to\C$ is a continuous $1$-periodic function then 
$x\mapsto f(\frac{\ln e^{i\langle k,x\rangle}}{2\pi i})$ is a continuous function which moreover belongs to $C_\pe(\R^d)$ provided $k\in \Et$. Hence, 
for $k_1,\cdots, k_n \in \Et$ the function
$$\textstyle x\mapsto b^{(n)}_{k_1,\cdots, k_n}(x) := b^{(n)}(\frac{k_1 x}{2\pi},\cdots,\frac{k_n x}{2\pi})$$
belongs to $M_m(C_\pe(\R^d))$.
\begin{theorem}\label{thm-diagram-1} 
Let $n\leq d$. The map $\Lambda^n \varphi : \Lambda^n\Et \to K_n(C_\pe(\R^d))$
$$ \Lambda^n \varphi(k_1\wedge \cdots \wedge k_n) =  [b^{(n)}_{k_1,\cdots, k_n}]_n$$
is a well-defined homomorphism of groups. Moreover,
\begin{equation}
\begin{array}{ccc}
\Lambda^n\Et & \stackrel{\Lambda^n\varphi}\to &K_n(C_\pe(\R^d))\\
 \downarrow  (2\pi)^{-n} \Lambda^n\imath    &                     & \downarrow \langle\ch^{(n)}_{\mean,\alpha},\cdot\rangle\\       
 \Lambda^n{\R^d}^* & = & \Lambda^n{\R^d}^*
 \end{array}
 \end{equation}
 commutes.
\end{theorem}
\begin{proof}
The coefficients of $\langle \ch^{(n)}_{\mean,\alpha},[b^{(n)}_{k_1,\cdots,k_n}]_n)\rangle $ in the basis $e^1\wedge\cdots\wedge e^n$ are given by
 $$\langle \ch^{(n)}_{\mean,\alpha},[b^{(n)}_{k_1,\cdots,k_n}]_n)\rangle_{i_1,\cdots,i_n} 
 = c_n \sum_{\sigma \in S_n} 
\mathrm{sgn}(\sigma)\mean \Tr  \left(b_{\underline{k}} \frac{\partial b_{\underline{k}}}{\partial x_{i_{\sigma(1)}}}  \cdots \frac{\partial b_{\underline{k}}}{\partial x_{i_{\sigma(n)}}}  \right)$$
where we abbreviated $\underline{k} = k_1,\cdots,k_n$.
Denote by 
$k_{ij}$ the $j$-coefficient of $k_i$ w.r.t.\ the basis $e^1,\cdots,e^d$ of ${\R^d}^*$. Then
$$ \frac{\partial b_{\underline{k}}(x)}{\partial x_{i_{\sigma(1)}}}\cdots \frac{\partial b_{\underline{k}}(x)}{\partial x_{i_{\sigma(n)}}}
= \frac1{(2\pi)^n}\sum_{j_1,\cdots,j_n = 1}^n 
(k_{j_1 i_{\sigma(1)}}\cdots k_{j_n i_{\sigma(n)}}) 
\partial_{j_1} b_{\underline{k}}(x)\cdots \partial_{j_n} b_{\underline{k}}(x)
$$
If we take the signed sum over all permutations $\sigma$, then all termes in which two of the $j_i$ are equal drop out so that 
\begin{eqnarray*}
\sum_{\sigma\in S_n} \mathrm{sgn}(\sigma)\frac{\partial b_{\underline{k}}(x)}{\partial x_{i_{\sigma(1)}}}\cdots \frac{\partial b_{\underline{k}}(x)}{\partial x_{i_{\sigma(n)}}}
&=& \sum_{\sigma,\tau \in S_n} \frac{\mathrm{sgn}(\sigma)}{(2\pi)^n} 
(k_{\tau(1) i_{\sigma(1)}}\cdots k_{\tau(n) i_{\sigma(n)}}) 
\partial_{\tau(1)} b_{\underline{k}}(x)\cdots \partial_{\tau(n)} b_{\underline{k}}(x)
\\
& = &   \frac{\det(k_1,\cdots,k_n)_{i_1,\cdots,i_n}}{(2\pi)^n}
\sum_{\tau \in S_n} \mathrm{sgn}(\tau) \partial_{\tau(1)} b_{\underline{k}}(x)\cdots \partial_{\tau(n)} b_{\underline{k}}(x)
\end{eqnarray*}
The expression vanishes if $k_1,\cdots, k_n$ are linear dependent in ${\R^d}^*$. Otherwise,
$$ \mean \Tr(b_{\underline{k}}(x)\partial_{\tau(1)} b_{\underline{k}}(x)\cdots \partial_{\tau(n)} b_{\underline{k}}(x)) = 
\mean \Tr(b_{\underline{e}}(x)\partial_{\tau(1)} b_{\underline{e}}(x)\cdots \partial_{\tau(n)} b_{\underline{e}}(x)) $$
with $\underline e = e^1,\cdots, e^n$,
as the trace per unit volume of a periodic function is invariant under an invertible linear change of variables. We hence find
$$c_n \sum_{\sigma\in S_n} \mathrm{sgn}(\sigma) 
\mean \Tr (b_{\underline{k}}(x)\partial_{\tau(1)} b_{\underline{k}}(x)\cdots \partial_{\tau(n)} b_{\underline{k}}(x)) =  \ch^{(n)}_{\mean,\alpha},[b^{(n)}_{\underline e}]_n)\rangle = 1$$
by (\ref{eq-norm})
and the fact that $x \mapsto b^{(n)}_{\underline e}(x_1,\cdots,x_n)$ is constant in directions perpendicular to $e^1,\cdots,e^n$.
We thus have shown
that $$\langle \ch^{(n)}_{\mean,\alpha},[b^{(n)}_{k_1,\cdots,k_n}]_n)\rangle_{i_1,\cdots,i_n} = \frac1{(2\pi)^n} \Lambda^n\imath(k_1\wedge\cdots \wedge k_n).$$ 
In particular, the map $(k_1,\cdots,k_n)\mapsto \langle \ch^{(n)}_{\mean,\alpha},[b^{(n)}_{k_1,\cdots,k_n}]_n)\rangle$ is additive in each variable and totally antisymmetric.

It remains to show that already $(k_1,\cdots,k_n)\mapsto [b^{(n)}_{k_1,\cdots,k_n}]_n)$ is additive in each variable and totally antisymmetric and thus $\Lambda^n\varphi$ a well-defined homomorphism of groups. Let $\Gamma\subset \Et$ be a sublattice of rank $r$. The $C^*$-algebra $C^*(\Gamma)$ generated by the $\chi_k$ with $k\in\Gamma$ is a subalgebra of $C_\pe(\R^d)$ and the diagram 
$$\begin{array}{ccc}
\Lambda^n\Gamma & \stackrel{\Lambda^n\varphi}\to &K_n(C^*(\Gamma))\\
 \downarrow &                     & \downarrow \jmath_*\\   
 \Lambda^n\Et & \stackrel{\Lambda^n\varphi}\to &K_n(C_\pe(\R^d))
 \end{array}$$
commutes. Here the left vertical arrow is induced by the inclusion of $\Gamma$ in $\Et$  and the right vertical the map on $K$-theory induced by the inclusion $\jmath:C^*(\Gamma)\to C_\pe(\R^d)$. As any two elements of $\Et$ lie in some finite rank sublattice $\Gamma$ of $\Et$ it is sufficient to verify additivity and antisymmetry of the map $\Lambda^n\varphi$ in the first line. Now $C^*(\Gamma)$ is isomorphic to the algebra of continuous functions on the dual of $\Gamma$ which is an $r$ dimensional torus. It is well known that the elements of the $K$-group of the torus algebra can be separated by their pairing with cyclic cohomology, that is, by the values $\ch^{(s)}$, $s\leq r$ can take on them. In other words,   $\Lambda^n\varphi$ is additive in each variable and antisymmetric if, for any sublattice $\Gamma\subset \Et$ of rank $r$ and $s\leq r$, the restriction of $\langle\ch^{(s)}, \Lambda^n\varphi (\cdot)\rangle$ to $\Lambda^n\Gamma$ is additive in each variable and totally antisymmetric. For $s=n$ we have shown this above. For $s<n$ the function vanishes identically, as $\langle \ch^{(s)},[b^{(n)}_{\underline e}]_n\rangle$ is known to be $0$ for the torus if $s<n$. 
\end{proof}

\subsection{The second commuting diagram and main result}
To obtain the second commuting square we apply
Theorem~\ref{thm-diagram-2} to $A=C_\pe(\R^d)$ with $\alpha$ the translation action and $\Tt=\mean$, the trace per unit volume. This yields the commuting diagram
\begin{equation} \label{eq-diagram-2}
\begin{array}{ccc}
 K_{n}(C_\pe(\R^d)) & \stackrel{\phi_\alpha^e} \to & K_{n+d}(C_\pe(\R^d)\rtimes_\alpha\R^d)  \\
  \downarrow  \langle\ch^{(n)}_{\mean,\alpha}   ,\cdot\rangle   &  & \downarrow 
 \epsilon_{d-n} \langle\ch^{(d-n)}_{\tp,\hat\alpha},\cdot\rangle  \\       
\Lambda^n{\R^d}^*  & \stackrel{D^e}\to &  \Lambda^{d-n}\R^d
 \end{array}
 \end{equation}
and when combined with Theorem~\ref{thm-diagram-1} proves our main result:
\begin{theorem}\label{thm-main-diagram} Let $\pe$ be a decorated Delone subset of $\R^d$ and $C_\pe(\R^d)$ the algebra of continuous functions which are pattern equivariant w.r.t.\ $\pe$. Let $\Et$ be the set of wave vectors $k\in{\R^d}^*$ for which the function $x\mapsto e^{i\langle k,x \rangle }$ belongs to $C_\pe(\R^d)$. Let $0\leq n\leq d$. The following diagram is commutative.
$$\begin{array}{ccc} 
\Lambda^n\Et  & \stackrel{\phi^e_\alpha\circ\Lambda^n\varphi} \to & K_{n+d}(C_\pe(\R^d)\rtimes_\alpha\R^d)  \\
  \downarrow  \frac1{(2\pi)^n}\Lambda^n\imath      &  & \downarrow 
 \epsilon_{d-n}\langle \ch^{(d-n)}_{\tp,\hat\alpha},\cdot\rangle  \\       
\Lambda^n{\R^d}^*  & \stackrel{D^e}\to &  \Lambda^{d-n}\R^d
 \end{array}$$ 
\end{theorem}
As already mentioned, if $n=d$ then the right column corresponds to the gap-labelling, as  $\ch^{(0)}_{\tp,\hat\alpha}$ is the $0$-cocycle given by the trace per unit volume $\tp$. We hence obtain, for models defined by operators $H$ which are bounded from below and whose resolvent $(H+i)^{-1}$ belongs to $\Ape$: 
 \begin{cor} \label{cor-diagram-1}
The gap labelling group $\gap$ associated to $\pe$ contains the group which is generated by $\frac1{(2\pi)^d}$ times the volumes of parallel epipeds in ${\R^d}^*$ which are spanned by vectors from $\Et$. 
\end{cor}
When considering a tight binding model on a set $\se\subset \R^d$ which is locally derivable from $\pe$ the gap labelling group has normalised by the density $\mathrm{dens}(\se)$ of this set. As a consequence we get:
\begin{cor} For tight binding models defined on $\se$,
the gap labelling group $\gapLS$ contains the subgroup which is generated by $\frac1{(2\pi)^d \mathrm{dens}(S)}$ times the volumes of parallelepipeds in $\hat\R^d$ which are spanned by vectors from $\Et$. 
\end{cor}
\begin{proof}
We can extend the commuting diagram of Theorem~\ref{thm-main-diagram} (with $n=d$) to the right by the diagram (\ref{eq-ME}).
\end{proof}
Let us mention that the statement of these corollaries is trivial if $\pe$ does not have the Meyer property, or is not topologically conjugate to a pattern which has the Meyer property \cite{kellendonk2014meyer}. Indeed, in this case
$\Et$ does not contain $d$ linear independent vectors and so there are no non-degenerate parallelepipeds. 

\subsection{External magnetic fields}
In the presence of external magnetic fields the observable algebra is the twisted crossed product 
$C_\pe(\R^d)\rtimes_{\alpha,\sigma_B} \R^d$, the twisting cocycle $\sigma_B$ depending on the magnetic field $B$ \cite{bellissard1986k}. At least for a constant magnetic field the twisted  crossed product may still be written as an iterated crossed product 
$$C_\pe(\R^d)\rtimes_{\alpha,\sigma_B} \R^d=C_\pe(\R^d)\rtimes_{\tilde\alpha_1}\R\rtimes_{\tilde\alpha_2}\R\cdots \rtimes_{\tilde\alpha_d}\R$$ so that the $K$-theory of the observable algebra is still isomorphic to that of $C_\pe(\R^d)$. However, the actions $\tilde\alpha_i$ are no longer translation actions, but rather magnetic translation actions, and therefore the corresponding Chern cocycles are different. For this reason the diagram which is analogous to (\ref{eq-diagram-2}) no longer makes sense. We give an example of this.
 
We consider a two dimensional system with constant perpendicular magnetic field $B$, like in the Quantum Hall effect, and use the Landau gauge.
Then the observable algebra can be written 
$C_\pe(\R^2)\rtimes_{\alpha_1}\R\rtimes_{\tilde\alpha_2}\R$ where $\alpha_1$ is the usual translation in the 1-direction, $\alpha_1^t(f)(x) = f(x+te_1)$, $f\in C_\pe(\R^2)$ while $\tilde\alpha_2$ is given by
$$\tilde\alpha_2^t(F)(s) = e^{iBts} \alpha^t_2(F(s))=\hat\alpha_1^{Bt}\alpha^t_2(F)(s)$$
on $F:C_c(\R, C_\pe(\R^2))$, where $\alpha_2$ is translation in the 2-direction on $C_\pe(\R^2)$.
It follows that the derivation corresponding to $\tilde\alpha_2$ picks up an extra term
$$\tilde\nabla_2(F)(s) = \nabla_2(F(s))-\frac1{2\pi}iBsF(s).$$  
This can be written
$$ \tilde\nabla_2 = \nabla_2+B\hat\nabla_1.$$
For non-zero $B$ we cannot express  $\langle \Ch_{\hat\mean},\phi_{\tilde\alpha_2}\phi_1 [x]_i\rangle$ through a Chern-cocycle defined on $C_\pe(\R^2)$ as in (\ref{eq-diagram-2}), since the action $\tilde\alpha_2$ is not defined on that algebra. Instead we can apply Theorem~\ref{thm-SS} to the algebra $A=C_\pe(\R^2)\rtimes_{\alpha_1}\R$ with $\R$-action $\tilde\alpha_2$ and the trace $\tau$ which is dual to $\mean$ on the crossed product with only the 1-component the action $\alpha_1$.  We obtain
\begin{eqnarray*}
\langle \Ch_{\hat\mean},\phi_{\tilde\alpha_2}[x]_i\rangle
& = & -\langle \Ch_{\hat\mean,\tilde\alpha_2,\hat{\tilde\alpha}_2},\phi_{\tilde\alpha_2}[x]_i\rangle\\
& = & -\langle \Ch_{\hat\mean,\alpha_2,\hat{\tilde\alpha}_2},\phi_{\tilde\alpha_2}[x]_i\rangle
-B\langle \Ch_{\hat\mean,\hat\alpha_1,\hat{\tilde\alpha}_2},\phi_{\tilde\alpha_2}[x]_i\rangle\\
& = & -\langle \Ch_{\tau,\alpha_2},[x]_i\rangle
-B\langle\Ch_{\tau,\hat\alpha_1},[x]_i\rangle
\end{eqnarray*}
This has the following physical interpretation which is best seen in the context of the Wiener-Hopf extension  describing the  
bulk boundary correspondence when a boundary is introduced along the $1$-direction. In this framework $C_\pe(\R^2)\rtimes_{\alpha_1}\R\otimes \mathcal K$ is the boundary algebra, $C_\pe(\R^2)\rtimes_{\alpha_1}\R\rtimes_{\tilde\alpha_2}\R$ the bulk-algebra and the exponential map $\exp$ of the Wiener Hopf extension equals $-\phi_{\tilde\alpha_2}^{-1}$. Given a bulk Hamiltonian $H$ with a gap at energy $E$ and associated spectral projection $P_E$ 
the above equation becomes
$$\langle \Ch_{\hat\mean},[P_E]_0\rangle
= \langle \Ch_{\tau,\alpha_2},\exp([P_E]_0)\rangle +
B\langle\Ch_{\tau,\hat\alpha_1},\exp([P_E]_0)\rangle
$$
The number on the left is the integrated density states of $H$ at energy $E$.
The first number on the right is the gradient pressure per unit energy exhibited on the boundary states in the gap, and the second number is $B$ times the edge conductivity \cite{kellendonk2005gap}.

\section{Topological eigenvalues, cohomology and the Ruelle Sullivan map}
\label{sec-coh}
The $K$-theory of a topological space is rationally isomorphic to its Cech cohomology and so there is an analogous description of the commuting diagram of Theorem~\ref{thm-diagram-1} in which the right column is replaced by Cech cohomology and the Ruelle Sullivan map from \cite{kellendonk2006ruelle}. We recall the relevant results from \cite{kellendonk2006ruelle} and \cite{barge2012maximal} and relate them to the work of \cite{akkermans2021relating}.

The results of \cite{kellendonk2006ruelle,barge2012maximal,akkermans2021relating} are formulated in the framework of $\R^d$ actions on topological spaces and here applied to $(\Omega_\pe,\alpha,\R^d,\mu)$ which is dual to $(C_\pe(\R^d),\alpha,\R^d,\mean)$. As 
$L^2(\Omega_\pe,\mu)$ (the dual to $L^2_\pe(\R^d,\mean)$) is separable and eigenfunctions to distinct eigenvalues are orthogonal, $\Et$ is a countable subgroup of ${\R^d}^*$ and thus the direct limit of free subgroups of finite rank. We equip it with the discrete topology. Then its Pontraryagin dual $\widehat \Et$ is an inverse limit of finite dimensional tori. Moreover, the exterior algebra $\Lambda^n \Et$ is naturally isomorphic to the Cech cohomology $\check H^n(\widehat{\Et},\Z)$ \cite{barge2012maximal}. 

Equipped with its induced action from $\R^d$ the dual group $\widehat \Et$ is the maximal equicontinuous factor of the dynamical system $(\Omega_\pe,\alpha,\R^d)$ \cite{aujogue2015equicontinuous,barge2012maximal}. The corresponding factor map $\pi:\Omega_\pe\to\widehat \Et$ induces a morphism $\pi_*:\check H^n(\widehat{\Et},\Z)  \to \check H^n(\Omega_\pe,\Z)$ which can be composed with the isomorphism between $\Lambda^n \Et$ and $H^n(\Omega_\pe,\Z)$ to obtain a group homomorphism
$$\Lambda^n\check\varphi:\Lambda^n \Et \to \check H^n(\Omega_\pe,\Z). $$
We may combine this homomorphism with the 
Ruelle Sullivan map $$\tau_{\alpha,\mu}:\check H^n(\Omega_\pe,\Z)\to \Lambda^n{\R^d}^*$$ from \cite{kellendonk2006ruelle}:
$\tau_{\alpha,\mu}$ is a composition of the morphism mapping the integer valued Cech cohomology $ \check H^n(\Omega_\pe,\Z)$ into the tangential cohomology $H^n_{tg}(\Omega_\pe,\R)$ of Moore and Schochet \cite{moore2006global} (the action on $\Omega_\pe$ is locally free), and then integrating against the measure $\mu$. 
As a result we get the following cohomological version of Theorem~\ref{thm-main-diagram}.
\begin{prop}
The following diagram is commutative
 \begin{equation}\label{diagram-1}
 \begin{array}{ccc}
 \Lambda^n\Et & \stackrel{\Lambda^n\check\varphi}\to & \check H^n (\Omega_\pe,\Z)  \\
 \downarrow    \frac1{(2\pi)^n} \Lambda^n\imath &                          & \downarrow \tau_{\alpha,\mu}  \\       
  \Lambda^n {\R^d}^* & = &  \Lambda^n {\R^d}^*   
 \end{array}
\end{equation} 
\end{prop}
 \begin{proof}
Since the image of $\Lambda^n\check\varphi$ lies in $\check H^n (\widehat{\Et},\Z)$ and $\widehat{\Et}$ is an inverse limit of finite dimensional tori on which the action restricts to a constant flow action, we can apply the same calculation as in the proof of Thm.~13 of \cite{kellendonk2006ruelle} to obtain the result.
 \end{proof}

In \cite{akkermans2021relating} the authors establish that $K_0(C(\Omega_\pe)\rtimes_\alpha \R^d)$ maps to the top degree tangential cohomology $H^d_{tg}(\Omega_\pe,\R)$, and that the diagram
 \begin{equation}\label{eq-diagram-3}
 \begin{array}{ccc}
K_{0}(C(\Omega_\pe)\rtimes_\alpha\R^d) ) & \to & H^d _{tg}(\Omega_\pe,\R) \\
 \downarrow \tau_*&                         & \downarrow   \int d\mu \\       
 \R & \cong  & \Lambda^d {\R^d}^* 
 \end{array}
 \end{equation}
 commutes (Thm.~9.1 of \cite{akkermans2021relating}). Here the right vertical arrow means integrating a form representing the cohomology class against the measure $\mu$ over $\Omega_\pe$, and the functional $\tau_*$ is defined by the dual trace $\tau$ on $C(\Omega_\pe)\rtimes_\alpha \R^d$. If $d\leq 3$ and $\pe$ has finite local complexity this can be strengthened to integer cohomology ({Cor.~9.5} of \cite{akkermans2021relating}), and then can be seen as the cohomological analog of (\ref{eq-diagram-2}) with $n=d$. To summarize, if $\pe$ has finitely local complexity and $d\leq 3$ then the two commuting diagrams  
(\ref{diagram-1}) and (\ref{eq-diagram-3}) combine to yield an alternative 
 proof of Theorem~\ref{thm-main-diagram} for $n=d$. 

 \section{The one dimensional situation} \label{sec-one-dim}
 The one dimensional situation is particularly interesting, as one dimensional Schr\"odinger operators are expected to have many gaps. Generically one can expect that all gaps predicted by the gap-labelling theorem are open. We therefore have a closer look at this case.
\subsection{Perturbative arguments}
If the solid is crystalline, that is, there is a lattice $\Gamma\subset \R$ such that $\pe_a+x=\pe_a$ for all $x\in\Gamma$, $a\in \al$,  then its Bragg peaks are located on $\Gamma^{rec}$, the reciprocal lattice in $k$-space \cite{ashcroft1976solid,baake2013aperiodic}. Likewise, the Fourier transform $\hat V$ of a $\Gamma$-periodic potential $V$ is supported on $\Gamma^{rec}$. The spectrum of the free Laplacian $ -\partial^2_x$ is $\{k^2:k\in \hat\R\}$ thus twofold degenerated if $k\neq 0$. 
In formal degenerate perturbation theory one finds that the degeneracy at energy $k^2$ can be lifted by the potential if $\hat V(2 k)\neq 0$, and as a consequence $H=-\partial^2_x + \lambda V$ may develop a gap at that energy if $k\in \frac{1}2 \Gamma^{rec}$ when $\lambda$ is turned on \cite{ashcroft1976solid,reed1978iv}.  This leads to a relation between the Bragg spectrum and the gaps in the electronic spectrum which is particularly simple in one dimension, as all gaps open up under the perturbation with a generic potential \cite{reed1978iv} and so $H$ has a band spectrum in which each band contributes equally to the integrated density of states. We can then count the bands from low to high energy either by the positive part of the Bragg spectrum--for that we need to choose a positive direction on $\Gamma^{rec}$--or by the values of the integrated density of states up to the gap. 
That these two ways of counting are the same is the content of Cor.~\ref{cor-diagram-1}, as a separated band has an integrated density of $1$ per unit cell. 

We should mention that the perturbation argument of \cite{reed1978iv} hinges on two crucial properties of the system: first, that the Hilbert space can be fibered over a parameter space in such a way that its fibres are preserved under the action of $H$, and second, that those degenerate subspaces of the restriction of $H$ to a fibre, whose degeneracies are to be lifted by the potential,  
are energetically isolated from their ortho-complement in the fibre. This is the case for periodic potentials for which the Hilbert space fibres over the Brillouin zone. But we cannot expect such a scenario for aperiodic potentials. It seems nevertheless compatible with numerical simulations that when $\hat V(2k)$ is large then a gap opens pertubatively in $\lambda$ at energy $k^2$, whether $V$ is periodic or not.

Arguments involving a periodic approximation of aperiodic potentials usually have a heuristic component, namely when it comes to the convergence gap labels. The difficult part here is that more and more gaps appear when the period grows and despite the fact that we usually have convergence of the spectra in the Hausdorff topology \cite{beckus2018spectral}, it is difficult to control which gaps converge to the gaps of the aperiodic operator. 

In \cite{gambaudo2014brillouin} the authors propose to approximate the potential $V$ by suppressing in its Fourier transform all but the strongest components, of which there ought to be only finitely many. In this way, $V$ is approximated by an almost periodic function with finitely many possibly incommensurate wave lengths (the chosen Bragg peak positions). In one dimension one can then apply the gap labelling by means of the rotation number of Johnson-Moser \cite{johnson1982rotation}
to obtain the gap labels of the approximated Hamiltonian. The gap-labelling theorem of Johnson-Moser states that these belong to the $\Z$-module generated by the $\frac1{2\pi}$ times the wave lengths. This is in agreement with Cor.~\ref{cor-diagram-1}, or, put differently, Cor.~\ref{cor-diagram-1} confirms that the heuristic approach to approximate the potential by an almost periodic one gives the right answer.

\subsection{$K$-theory arguments} Let us describe again the content of Theorem~\ref{thm-diagram-2} in one dimension where all the maps involved can be made very explicit. 
We saw that $k$ is a topological eigenvalue if and only if the function
$\chi_k:\R\to \C$, $\chi_k(x) =e^{i\langle k,x\rangle}$ is an element of $C_\pe(\R)$. 
As $\chi_k$ is a unitary it defines an element of  $K_1(C_{\pe}(\R))$ which allowed us to define $\varphi:\Et \to K_1(C_{\pe}(\R))$,
$\varphi(k) =  [\chi_k]_1$.
The equation $\chi_{k_1+k_2} = \chi_{k_1}\chi_{k_2}$ confirms that this map is a group homomorphism. 
The $1$-cocycle $\Ch_{\mean,\alpha}$ entering in Theorem~\ref{thm-diagram-2} is given by
$(f_0,f_1) \mapsto \frac1{2\pi i} \mean (f_0 \partial_{x} f_1) $. Thus 
$$\langle\Ch_{\mean,\alpha},[\chi_k]_1\rangle =  \frac1{2\pi i} \mean (\chi_k^{-1} \partial_{x}\chi_k) = \frac{k}{2\pi}$$ 
which is equal to $1$ if $k$ is a generator of the lattice reciprocal to $\Z$, $k = 2\pi$. 

The Connes Thom isomorphism $\phi_\alpha$ maps $K_1(C_{\pe}(\R))$ to $K_0(C_{\pe}(\R)\rtimes_\alpha\R)$ but its general form is not very explicit. In our context we can use  the commuting diagram (\ref{eq-diagram-2}) to
compute its values on $[\chi_k]_1$, for any $k\in\Et$. 
For that consider the Mathieu Hamiltonian 
$H_k(\lambda) = -\partial_x^2 + \lambda (\chi_k+\bar\chi_k)$ with coupling constant $\lambda>0$.
It is known that all possible gaps of $H_k(\lambda)$ are open \cite{reed1978iv}. 
Let $P_k(\lambda)$ be the spectral projection associated with the first gap of $H_k(\lambda)$, that is, the spectral projection onto the lowest energy band of $H_k(\lambda)$. 
Given that $H_k(\lambda)$ is bounded from below and
$\chi_k$ an element of $C_\pe(\R)\rtimes_\alpha\R$, $P_k(\lambda)$ is a projection of 
$C_\pe(\R)\rtimes_\alpha\R$. It depends continuously on $\lambda$ and thus its $K_0$-class is independent of $\lambda$ as long as $\lambda>0$. 
The integrated density of states at the first gap of $H_k(\lambda)$ is $\tp(P_k(\lambda))$. It can be calculated in the limit $\lambda\to 0$ to be $\frac{|k|}{2\pi}$. 
It follows that $\tp_*([P_k(\lambda)]) = \frac{|k|}{2\pi}$ thus coinciding with $\langle\Ch_{\mean,\alpha},[\chi_k]_1\rangle$ if $k$ is positive. 
We claim that this implies that $\phi_\alpha([\chi_k]_1) = [P_k(\lambda)]_0$.
  
To proof the claim we observe that  the $C^*$-algebra generated by $\chi_k$ is $C_{\frac{2\pi}{k}\Z}(\R)$, the algebra of continuous $\frac{2\pi}{k}$-periodic functions. It is a subalgebra of $C_\pe(\R)$. 
Its $K_1$-group is generated by $[\chi_k]_1$. 
By naturality of the Connes Thom isomorphism 
$\phi_\alpha([\chi_k]_1)$ is determined by its image in $K_0(C_{\frac{2\pi}{k}\Z}(\R)\rtimes_\alpha\R)$. 
As $K_0(C_{\frac{2\pi}{k}\Z}(\R)\rtimes_\alpha\R)\cong \Z$ the tracial state $\tp_*$ must be injective on $K_0(C_{\frac{2\pi}{k}\Z}(\R)\rtimes_\alpha\R)$. We just saw that 
$\tp_*([P_k(\lambda)]_0) = \frac{|k|}{2\pi}$.
By Corollary~\ref{cor-diagram-1} $\frac{|k|}{2\pi}$ is a generator of the gap labelling group associated to $\pe = \frac{2\pi}{k}\Z$. Hence $[P_k(\lambda)]_0$ is a generator of  $K_0(C_{\frac{2\pi}{k}\Z}(\R)\rtimes_\alpha\R)$ so that 
$$\phi^1_\alpha([\chi_k]_1) = \mathrm{sign}(k) [P_k(\lambda)]_0.$$

\subsection{When is the Gap-Labelling group determined by the topological Bragg spectrum?}
Now that we have defined a group homomorphism from $\Lambda^d\Et$ to $\gap$ we ask the question: when is this homomorphism surjective? For brevity we call this homomorphism the Bragg to Gap map. 

If the dimension is at most $3$ and $\pe$ has finite local complextiy (and there is no external magnetic field\footnote{With constant external magnetic field the situation is more complicated, as the magnetic field accounts for extra generators \cite{benameur2020proof}.})
 then the gap-labelling theorem tells us that $\gapLS$ must be the frequency module of $\pe$.  This can be employed to calculate $\gapLS$ and $\gap$ in a variety of situations, as for instance for substitution tilings and for almost canonical cut-and-project patterns. In both cases all dynamical eigenvalues are topological and, in particular for cut and project patterns, $\Lambda^d\Et$ is easy to compute, see for instance \cite{kellendonk2006ruelle}. In the latter case the question after the surjectivity comes down to calculating the volumes of acceptance domains. For substitution tilings, $\gap$ can be computed from the data of the substitution \cite{kellendonk2000tilings}. 
 Explicit calculations show that for Penrose tilings and for octagonal tilings, 
$(2\pi)^{-d}\det(\Lambda^d\Et)$ is a finite index subgroup of $\gapLS$ \cite{kellendonk1995noncommutative}.

For one-dimensional systems one can say a little more. Using the isomorphism between
$\check H^1(\Omega_\pe,\Z)$ and $K_0(C_\pe(\R)\rtimes_\alpha\R)$ and comparing (\ref{diagram-1}) with Theorem~\ref{thm-diagram-2} the Bragg to gap map corresponds in cohomology to the composition
$$ \Et \stackrel{\check\varphi}\to  \check H^1(\Omega_\pe,\Z) \stackrel{\tau_{\alpha,\mu}}\to \R$$
which has been analysed in \cite{barge2012maximal}. 
Now $\gap$ is the image of $\tau_{\alpha,\mu}$. 
Recall that $\tau_{\alpha,\mu}\circ \check\varphi$ is injective.
We therefore have two exact sequences with $\check H^1(\Omega_\pe,\Z)$ in the middle
\begin{eqnarray} 0\to \Et \stackrel{\check\varphi}\to & \check H^1(\Omega_\pe,\Z) &
\to\mathrm{coker}\check\varphi \to 0 \label{SES-phi} \\
0\to \ker \tau_{\alpha,\mu} \to& \check H^1(\Omega_\pe,\Z) &
\stackrel{\tau_{\alpha,\mu}}\to \gap \to 0 \label{SES-RS} 
\end{eqnarray}
A diagram chase shows that if $\tau_{\alpha,\mu}\circ \check\varphi$ is also surjective, then both sequences must split and $\check H^1(\Omega_\pe,\Z)\cong  \Et \oplus \ker \tau_{\alpha,\mu}$.
Formulated in terms of $K$-theory this means that
$$ K_0(C_\pe(\R)\rtimes_\alpha\R) \cong \Et \oplus \ker \tp_* $$
provided the Bragg to gap map is surjective.

The following result from \cite{andress2016cech} fits into this context. It applies to one-dimensional substitution tilings. We associate to a tiling the Delone $\pe$ defined by its boundary points. The reader can consult \cite{andress2016cech} for the precise definitions of the involved substitutions and their associated notions. 
It is known that for such substitution tilings all dynamical eigenvalues are topological $\mathcal E_\pe=\Et$ \cite{solomyak2007eigenfunctions}. 
\begin{theorem}[\cite{andress2016cech}]
Consider a one-dimensional primitive substitution tiling with common
prefix. Assume furthermore that the substitution is irreducible. 
Then $\tau_{\alpha,\mu}$ is injective. 
Moreover, if its dilation factor is a Pisot number then $\check\varphi$ is surjective.
\end{theorem}
In particular, for a primitive irreducible substitution with common prefix whose dilation factor is a Pisot number, we have
\begin{equation}\label{eq-subst}
\mathcal E_\pe=\Et \stackrel{\phi_\alpha\circ \varphi}{\cong}  K_0(C_\pe(\R)\rtimes_\alpha\R) \stackrel{\tp_*}{\cong} \gap
\end{equation}
and the Bragg to gap map is bijective. The Fibonacci tiling is the most prominent example of such a substitution tiling.

It this context we point out that a one-dimensional primitive substitution tiling has a non-trivial group of eigenvalues if and only if the dilation factor is a Pisot number \cite{solomyak1997dynamics}.  So for primitive irreducible substitution tilings with common prefix the only alternative to (\ref{eq-subst}) is $\Et=\{0\}$ and in this case the Bragg spectrum does not give any information on the gap-labelling.

\subsection{Example: the Thue-Morse solid}
We end this section with an example in which the Bragg to gap map is not surjective.

The Thue-Morse sequence represents the spatial structure of a one-dimensional solid with two types of atoms $\al = \{1,\bar 1\}$ which is obtained by the substitution
$$1 \mapsto 1 \bar 1,\quad \bar 1 \mapsto \bar 1 1  .$$
This is a primitive reducible Pisot-substitution with scaling factor $2$. 
More precicely, the Thue Morse sequence is a twosided sequence $(\sigma_n)_{n\in \Z}$ which is a fixed point of the above substitution. It can be recursively defined as follows:  let $s_0=1$ and
$$s_{n+1} = s_n \overline{s_n}$$
for $n\geq 0$ 
where $\overline{a_1 a_2 \cdots a_k} = \bar a_1 \bar a_2 \cdots \bar a_k$ and $\bar{\bar 1} = 1$. 
In the limit $n\to \infty$ one obtains a one-sided infinite sequence with values in $\al$ which we denote $(\sigma_n)_{n\in\NM}$. This sequence is then completed to a two-sided infinite sequence $(\sigma_n)_{n\in\Z}$
by mirror reflection: $\sigma_{-n} = \sigma_{n-1}$. Finally the subsets 
$$\pe_1=\{n\in\Z: \sigma_n =1 \}, \quad \pe_{\bar 1}=\{n\in\Z: \sigma_n = \bar{1} \}$$ model the spatial structure of the solid.

The gap-labelling group for this solid was computed by Bellissard \cite{bellissard1990spectral}. In \cite{barge2012maximal} one can find the computation of the Bragg to gap map (on the level of cohomology) together with the result, that  the exact sequence (\ref{SES-RS}) does not split. Indeed, $\Et$ is mapped under the Bragg to gap map into a subgroup of index $3$ of the gap-labelling group
$$\frac1{2\pi}  \Et = \left\{\frac{m}{2^n} : m\in\Z,n\in\NM\right\} \stackrel{}
\hookrightarrow
\left\{{\frac13}\frac{m}{2^n} : m\in\Z,n\in\NM\right\} = \gap $$
For a generic choice of pattern equivariant potential all gaps allowed by the gap-labelling group are indeed open. But for the tight binding operator
$$H_\lambda \psi_n=\psi_{n+1}+\psi_{n-1} -2\psi_n + V(n)\psi_n $$
with $V(n) = 1$ if $n\in {\pe_1}$ and  $V(n) = -1$ if $n\in {\pe_{\bar 1}}$, there is a symmetry which prevents the gaps coming from $\Et$, except for the one with label $\frac12$, to open \cite{bellissard1990spectral}. 
This points out the question, what is the gap opening mechanism for the gaps of this operator? 

In early work predating the gap-labelling theorem of Bellissard, Luck presented heuristic arguments predicting that a gap opens perturbatively in the spectrum of $H_\lambda$ at those values for $k\in\hat \R$ for which the scaling exponent satisfies
$$\gamma(k) := \limsup_{N\to +\infty} \frac{\ln  \prod_{l=0}^{N-1} \sin^2(\pi 2^l k)}{\ln 2^N} > -1 $$
cf.\ the statement (a) on page 5838 combined with Equations (3.7,3.8) of \cite{luck1989cantor}. His arguments were based on the transfer map for $H_\lambda$ and a perturbative analysis of its associated Liapounov exponent. Luck's predictions are, however, contradicted by the gap-labelling theorem. In fact, there are infinitely many $q\in\NM$ for which $\gamma(\frac{2\pi}{q})> -1$, see \cite{baake2013scaling} where $\gamma$ is denoted by $\beta$, while the gap labelling group contains only denominators of the form $q=3\times 2^n$. This underlines the power of the gap-labelling theorem: it gives us a selection rule telling us which gaps can open under perturbation. We are not aware of any analytic arguments (not based on $K$-theory)  which explain why the gap at wave vector $k=\frac13$ opens while, for instance, at wave vector $k=\frac15$ it cannot open.
Of course, as $\gamma(\frac13)$ is quite much bigger than $\gamma(\frac15)$ one could expect the gap at $k=\frac13$ to be larger than the one at $k=\frac15$, but the gap labelling tells us that the latter cannot open at all.
The gap opening mechanism for gaps with labels from outside of $\Et$ remains mysterious.

\bibliography{Bragg2Gap}{}
\bibliographystyle{amsplain}

\end{document}